\documentclass[12pt]{article}
\usepackage{newtxtext,newtxmath}
\usepackage{graphicx}
\usepackage[letterpaper,margin=1in]{geometry}
\newtheorem{definition}{Definition}
\newtheorem{proposition}{Proposition}

\renewenvironment{abstract}
	{\quotation}
	{\endquotation}

\date{}

\makeatletter
\renewcommand{\fnum@figure}{\textbf{Figure \thefigure}}
\renewcommand{\fnum@table}{\textbf{Table \thetable}}
\makeatother

\usepackage{scicite}
\usepackage{url}
\usepackage{caption}

\def\scititle{
	SheafIQ: Sheaf-Theoretic Information Quantification of Vector Fields on Geometric Graphs
}

\title{\bfseries \boldmath \scititle}
\author{
	Cong Shen$^{1}$,
	Guancen Lin$^{1\ast}$,
	Chuan-Shen Hu$^{2\ast}$\and
	\small$^{1}$ State Key Laboratory of Mathematical Sciences, Academy of Mathematics and Systems Science, \and \small Chinese Academy of Sciences, 100190, Beijing, China.\and\
	\small$^{2}$Department of Applied Mathematics, National University of Kaohsiung, 81148, Kaohsiung, Taiwan.\and
	\small$^\ast$Corresponding author. Email: linguancen@amss.ac.cn, chuanshenhu1@nuk.edu.tw\and
}

\begin{document} 

\maketitle

\begin{abstract}
Vector fields on graph structures naturally arise in diverse biological and engineered systems, where vector-valued states are defined on the nodes and evolve through the network interactions. Existing methods primarily characterize either the graph topology or individual signals, but generally do not quantify how local interactions among node-associated vectors are organized across the graph. To address this limitation, a sheaf-theoretic framework, termed SheafIQ, is proposed to represent neighboring vectors in a common edge-associated coordinate system, map local incompatibilities to a residual energy distribution, and quantify its global organization through entropy. Across proteins, functional brain networks, urban traffic systems, and power grids, SheafIQ consistently reveals complementary organizational information beyond conventional graph- and signal-based descriptors. More broadly, it establishes a unified information-theoretic framework for quantifying the organization of vector-valued states on geometric graphs, extending network analysis beyond graph topology alone.
\end{abstract}


\section{Introduction}

Complex systems are commonly represented as networks, in which nodes denote interacting entities and edges encode their structural or functional relationships \cite{strogatz2001exploring, boguna2009navigability, ruths2014control, artime2024robustness}. Beyond the underlying network connectivity, many real-world systems also involve vector-valued states or signals associated with the nodes. For example, proteins exhibit collective residue displacement fields during conformational motions, brain networks support spatially coordinated patterns of neural activity, traffic systems evolve through dynamic flow fields constrained by road connectivity, and power grids operate under continuously varying electrical states \cite{zheng2003comparative, park2013structural, bassett2017network, herty2003modeling, pagani2013power}. In these systems, the network topology specifies where interactions occur, whereas node-associated vectors describe the spatial organization of physical or functional states over the network. Consequently, system behavior is determined not only by the network topology, but also by how neighboring vector-valued states interact in the presence of additional geometric information (e.g., node coordinates). Understanding this organization is therefore essential for revealing the principles underlying the emergence, propagation, and reorganization of coordinated behaviors in complex systems.

Existing network analysis methods provide a diverse collection of tools for characterizing complex systems, including structural descriptors, graph entropy measures, graph signal processing techniques based on spectral and variation formulations, and application-specific statistical analyses \cite{donetti2005entangled, shen2025torsion, dehmer2011history, alon1996source, shen2026chi, sandryhaila2013discrete}. Meanwhile, many application-specific studies characterize node or edge states using statistical quantities such as signal magnitudes, pairwise differences, and global energy \cite{sandryhaila2014discrete, motsch2014heterophilious, yang2009protein}. These approaches have substantially advanced the analysis of network structure and graph-supported data. However, they primarily characterize either the underlying graph topology or the attributes of individual signals, without explicitly quantifying how local vector-valued interactions are spatially organized over the graph. Consequently, vector fields defined on the same underlying graph may exhibit fundamentally different organizations of local vector coordination, even when they share identical topology and similar overall signal statistics.

However, characterizing the organization of vector fields on graphs requires more than comparing signal magnitudes alone. In many systems, the physical interpretation of a vector is inherently determined by its local geometric context \cite{aminov2000geometry,de2016vector,crane2018discrete,qiu2022mapping,yao2022guiding}. For two neighboring nodes, vector differences parallel and perpendicular to their interaction may carry fundamentally different physical meanings. Consequently, directly comparing neighboring vectors in their original coordinate systems may fail to capture their underlying physical or functional compatibility. Instead, neighboring vectors should first be expressed in a common local reference frame associated with their interaction, allowing local compatibility to be distinguished from genuine local inconsistency. Once such local incompatibilities are identified, an equally important question arises: how are they organized across the entire network? Bridging local compatibility with global organization is essential for understanding the collective behavior of node states on complex networks.

To address these challenges, we introduce SheafIQ, a sheaf-theoretic framework for quantifying the local and global organization of vector fields on graphs. The central idea of SheafIQ is to bridge local vector compatibility and global organization of local incompatibilities within a unified framework. Locally, neighboring node states are represented in a common edge-associated coordinate system through a cellular sheaf, enabling their compatibility to be evaluated in a geometrically meaningful manner. Globally, the resulting sheaf residuals induce a residual energy field, whose normalized spatial distribution characterizes the organization of local incompatibilities over the graph. Rather than quantifying only the magnitude of local incompatibilities, SheafIQ measures their global organization through the entropy of this normalized residual energy distribution.

We demonstrate the generality and interpretability of SheafIQ across diverse biological and engineered systems. In protein systems, SheafIQ identifies mutation-induced dynamical reorganization, conformational hotspots, and long-range communication pathways. In functional brain networks, it reveals disease-associated organizations of local incompatibilities that remain reproducible across subjects and analysis settings. In urban traffic systems, SheafIQ characterizes large-scale spatiotemporal organization beyond conventional traffic descriptors, while in power grids, it captures localized electrical disturbances and identifies structurally critical components without requiring changes in network topology. Together, these results demonstrate that SheafIQ provides a unified and interpretable framework for characterizing vector fields on graphs across diverse application domains. More broadly, SheafIQ establishes a general information-theoretic framework for quantifying vector fields on graphs, extending network information analysis beyond graph topology to the organization of node states.

\section{Methods}

\subsection{Vector Fields on Geometric Graphs}

Most existing graph analysis methods characterize a network primarily through its topological structure, with nodes representing entities and edges representing pairwise relationships. Although this representation effectively describes structural connectivity, it is insufficient for systems in which each node also carries a vector-valued physical or functional state or signal. In many real-world systems, the graph topology specifies the interaction backbone, whereas the node-associated vectors describe how states, signals, forces, displacements, or flows are distributed over this backbone.

To formalize this setting, let $G=(V,E)$ be a finite, simple, undirected, and unweighted graph, where \(V\) is the node set and \(E\) is the edge set, and let $v\to e$ indicate that $v$ is an endpoint of $e$. When the nodes are indexed as $v_i$, an edge with endpoints $v_i$ and $v_j$ is denoted by $e_{ij}$. Each node \(v_i\in V\) is associated with a (column) vector $\mathbf{s}_i\in\mathbb{R}^{q}$ representing its local state or signal, called a \textit{node vector}. The collection of all node vectors is denoted by $\mathcal{S}=\{\mathbf{s}_i\}_{i\in V}$. We refer to the pair $(G,\mathcal{S})$ as a vector field on $G$, or simply a vector field when no ambiguity arises. In this formulation, the graph specifies the discrete pairwise relationships among the nodes, while the vector field describes the node vectors distributed over the graph.

This distinction is important because the same graph topology may give rise to very different organizations of node vectors. Conventional graph entropy measures, being determined solely by the connectivity pattern of \(G\), assign the same entropy value to graphs with identical topology, even when their node-associated vector fields differ substantially. However, in many systems, functional changes are primarily reflected in the organization of the underlying vector field rather than in the graph's topology. Therefore, the objective is not to quantify topology alone, but to characterize how node vectors are organized with respect to the underlying graph structure.

Rather than considering vector fields on graphs with purely combinatorial information, we further assume that the underlying graph is geometrically embedded in $\mathbb{R}^{d}$, where each node $i$ is assigned a coordinate $\mathbf{x}_i\in\mathbb{R}^{d}$. The coordinate \(\mathbf{x}_i\) provides a geometric reference for comparing vector quantities across adjacent nodes. In particular, for an edge \(e_{ij} = e_{ji} =\{ v_i, v_j \}\), the coordinates of its endpoints induce a local edge direction. This direction is essential for decomposing a node vector into its components parallel and orthogonal to the edge.

In general, the dimension $q$ of the node vectors need not equal the dimension $d$ of the geometric embedding. In the present setting, however, we assume $q=d$, so that each node vector can be decomposed into components parallel and orthogonal to the corresponding edge direction. 

Accordingly, the basic object studied in this work is not a graph alone, but a vector field on a geometric graph. The graph determines where local comparisons are made, the node coordinates provide a local geometric reference along each edge, and the node vectors specify the physical or functional quantities being compared. In particular, as developed in Section~\ref{Section: Cellular Sheaf Representation}, a sheaf-based framework is proposed to formulate the local compatibility between neighboring node vectors along graph edges. Building on this formulation, \emph{SheafIQ} (\emph{Sheaf-Theoretic Information Quantification}) provides an information-theoretic framework for quantifying the global organization of local incompatibilities (Section~\ref{Subsection: Information Quantification}), whose utility is demonstrated across diverse biological and engineered systems in Sections~\ref{Subsection: SheafIQ Reveals Mutation-Induced Dynamical Reorganization}--\ref{Subsection: Characterizing Power Grid Organization with SheafIQ}.
\subsection{Cellular Sheaf Representation}
\label{Section: Cellular Sheaf Representation}

To characterize how node vectors interact across a graph, we introduce a cellular sheaf defined on the underlying graph. Unlike a conventional graph, which only encodes adjacency relationships, a cellular sheaf additionally equips each cell with a local vector space and specifies how information is represented and compared between neighboring cells \cite{curry2014sheaves,curry2015tdacosheaf,curry2016discretemorse, hansen2019toward, battiloro2024tangent, bodnar2022neural,cooperband2024homology,Cooperband2023Homological}. Accordingly, the graph describes where interactions occur, whereas the sheaf specifies how vector-valued information is represented and compared across these interactions. 

\subsubsection{Cellular Sheaves on Geometric Graphs}

Let $G=(V,E)$ be a finite, simple, undirected, and unweighted graph with node set $V$ and edge set $E$. A \textit{cellular sheaf} of real vector spaces on \(G\) consists of the following data:
\begin{itemize}
\item[\rm (a)] \(\mathcal{F}(v)\) denotes the vector space attached to node \(v\);
\item[\rm (b)] \(\mathcal{F}(e)\) denotes the vector space attached to edge \(e\);
\item[\rm (c)] \(\mathcal{F}(v\rightarrow e):\mathcal{F}(v)\rightarrow\mathcal{F}(e)\) is an $\mathbb{R}$-linear transformation from $\mathcal{F}(v)$ to $\mathcal{F}(e)$.
\end{itemize}
The spaces \(\mathcal{F}(v)\) and \(\mathcal{F}(e)\) are called the \textit{stalk spaces} at the node \(v\) and the edge \(e\), respectively. For each endpoint \(v\) of an edge \(e\), the linear map \(\mathcal{F}_{v\rightarrow e}:\mathcal{F}(v)\rightarrow\mathcal{F}(e)\) is called the \textit{restriction map}, which maps the information in \(\mathcal{F}(v)\) to the stalk space \(\mathcal{F}(e)\). This data is summarized by the triple
\[
\mathcal{F}
=
\left(
\{\mathcal{F}(v)\}_{v\in V},
\{\mathcal{F}(e)\}_{e\in E},
\{\mathcal{F}(v\rightarrow e)\}
\right).
\]
Beyond the topological structure of the underlying graph, SheafIQ also incorporates its geometric embedding. Specifically, each node is assigned a coordinate in $\mathbb{R}^d$, and the node-associated vectors (i.e., states or signals) are assumed to lie in the same ambient space. That is, $q=d$.

Unlike node states (or signals), edge states are not meant to preserve the original vector representation. Instead, through the restriction maps, they provide a common computational space in which the states at the two incident nodes can be directly compared. Accordingly, the edge stalk space is defined by $\mathcal{F}(e)=\mathbb{R}^{2d}$ for each edge $e\in E$, where the two components correspond to the projections of the node states onto the directions parallel and orthogonal to the edge, respectively (see Section~\ref{subsubsection: Construction of Restriction Maps}).

Compared with a conventional graph, a cellular sheaf enriches the graph structure by introducing local vector spaces together with restriction maps. Consequently, neighboring node states are no longer compared solely in the ambient Euclidean space, but in the stalk space of their shared edge. This enables a more refined notion of local compatibility (see Section~\ref{subsubsection: Local Compatibility}).

\subsubsection{Construction of Restriction Maps}
\label{subsubsection: Construction of Restriction Maps}

To define a cellular sheaf on a graph $G=(V,E)$ with node set $V=\{v_1,\ldots,v_n\}$, we assume that each node $v_i$ is assigned a coordinate $\mathbf{x}_i\in\mathbb{R}^d$. For every edge $e_{ij}=\{v_i,v_j\}$, the restriction maps $\mathcal{F}(v_i\to e_{ij}):\mathcal{F}(v_i)\to\mathcal{F}(e_{ij})$ are defined for its two endpoints. The node coordinates naturally determine the direction vector of the relation $v_i \to e_{ij}$:
\[
\mathbf{u}_{ij}
=
\frac{\mathbf{x}_j-\mathbf{x}_i}
{\|\mathbf{x}_j-\mathbf{x}_i\|}.
\]

The (normalized) direction vector establishes a local reference frame for edge \(e_{ij}\). The restriction map is defined as an \emph{edge-induced coordinate transformation} that expresses node vectors in the local coordinate system determined by the edge. Under the assumption \(q=d\), the node vector and the edge direction belong to the same ambient space. For a node vector $\mathbf{s}_i\in\mathcal{F}(v_i)$, the restriction map $\mathcal{F}(v_i\to e_{ij})$ is represented by the $2d\times1$ matrix
\begin{equation}
\label{Eq. Matrix representation of the restriction map}
\mathcal{F}(v_i\to e_{ij})
(\mathbf{s}_i)
=
\begin{bmatrix}
(\mathbf{s}_i^{\top}\mathbf{u}_{ij})\mathbf{u}_{ij}
\\[1ex]
\mathbf{s}_i-
(\mathbf{s}_i^{\top}\mathbf{u}_{ij})
\mathbf{u}_{ij}
\end{bmatrix}
\in\mathcal{F}(e_{ij}),
\end{equation}
where the first block is the signed projection of $\mathbf{s}_i$ onto the edge direction, while the second block is its component orthogonal to the edge direction. Note that although $\mathbf{u}_{ij}=-\mathbf{u}_{ji}$, the corresponding restriction maps coincide, since the projection onto the line spanned by $\mathbf{u}_{ij}$ is invariant under reversing its orientation. That is,
\begin{equation}\label{Eq. F(vi to eij) = F(vj to eij)}
\mathcal{F}(v_i\to e_{ij}) =  \mathcal{F}(v_j\to e_{ij}).   
\end{equation}

Consequently, the restriction map transforms an arbitrary vector
$\mathbf{s}_i\in\mathbb{R}^{d}$ into a unified $2d$-dimensional edge representation,
\[
\mathcal{F}(v_i\to e_{ij}):
\mathbb{R}^{d}
\longrightarrow
\mathbb{R}^{2d},
\]
thereby enabling vectors associated with different nodes to be compared
within a common local coordinate system.

\subsubsection{Local Compatibility}
\label{subsubsection: Local Compatibility}

One fundamental principle of cellular sheaves is to capture the local compatibility among all node states (or signals) across neighboring cells while characterizing the global deviation of the entire state from harmonicity or compatibility \cite{hansen2019toward, hansen2019learning, robinson2020assignments, moustafa2026finite}. For an edge \(e_{ij} = \{ v_i, v_j \} \), the restricted vectors
\[
\mathcal{F}(v_i\to e_{ij})(\mathbf{s}_i),
\qquad
\mathcal{F}(v_j\to e_{ij})(\mathbf{s}_j)
\]
are regarded as locally compatible if
\[
\mathcal{F}(v_i\to e_{ij})(\mathbf{s}_i)
=
\mathcal{F}(v_j\to e_{ij})(\mathbf{s}_j).
\]
This condition indicates that the two node vectors exhibit identical local behavior when expressed in the coordinate system induced by the connecting edge. Conversely, whenever
\[
\mathcal{F}(v_i\to e_{ij})(\mathbf{s}_i)
\neq
\mathcal{F}(v_j\to e_{ij})(\mathbf{s}_j),
\]
the edge exhibits local incompatibility, indicating that the incident node vectors differ in their edge-oriented representations. This notion of local compatibility forms the foundation of the proposed framework. In the next section, local incompatibility is quantified through sheaf residuals, which provide a local measure of vector inconsistency on each graph edge.

\subsection{Sheaf Residuals}

The cellular sheaf introduced in the previous section provides a common local representation for neighboring node vectors. Once all node vectors have been expressed in the same edge-associated stalk space, local compatibility can be quantified continuously. To this end, we introduce the sheaf residual, which provides a quantitative measure of local inconsistency on each graph edge.

\subsubsection{Quantifying Local Incompatibility}

For an edge $e_{ij}=\{ v_i, v_j \}$, the restriction maps produce two edge representations,
\[
\mathcal{F}(v_i\to e_{ij})(\mathbf{s}_i),
\qquad
\mathcal{F}(v_j\to e_{ij})(\mathbf{s}_j),
\]
which belong to the same edge stalk space $\mathcal{F}(e)$.

If $\mathcal{F}(v_i\to e_{ij})(\mathbf{s}_i) = \mathcal{F}(v_j\to e_{ij})(\mathbf{s}_j)$, then the two node vectors are locally compatible on the edge \(e_{ij}\). However, exact local compatibility rarely occurs in practice. Instead of enforcing this equality, we define the difference between the two restricted vectors as the \emph{sheaf residual},
\begin{equation}\label{Eq. Sheaf residual}
\mathbf{t}_{ij} = \mathcal{F}(v_j\to e_{ij})(\mathbf{s}_j) - \mathcal{F}(v_i\to e_{ij})(\mathbf{s}_i).    
\end{equation}

Unlike the ordinary Euclidean difference \(\mathbf{s}_j-\mathbf{s}_i\), the sheaf residual is defined by comparing the two node vectors only after they have been transformed into a common edge-induced coordinate system. Consequently, it quantifies the local incompatibility of neighboring vectors relative to the underlying graph geometry.

\subsubsection{Parallel--Orthogonal Residual Decomposition}

Since the restriction map consists of parallel and orthogonal components~\eqref{Eq. Matrix representation of the restriction map}, the sheaf residual defined in~\eqref{Eq. Sheaf residual} admits the decomposition
\[
\mathbf{t}_{ij}
=
\begin{bmatrix}
\mathbf{t}_{ij}^{\parallel}
\\[1ex]
\mathbf{t}_{ij}^{\perp}
\end{bmatrix},
\]
where
\[
\mathbf{t}_{ij}^{\parallel}
=
(\mathbf{s}_j^{\top}\mathbf{u}_{ij} - \mathbf{s}_i^{\top}\mathbf{u}_{ij}) \mathbf{u}_{ij}
\]
denotes the mismatch along the edge direction, and
\[
\mathbf{t}_{ij}^{\perp} = (\mathbf{s}_j- (\mathbf{s}_j^{\top}\mathbf{u}_{ij}) \mathbf{u}_{ij}) - (\mathbf{s}_i - (\mathbf{s}_i^{\top}\mathbf{u}_{ij})
\mathbf{u}_{ij})
\]
measures the difference between the orthogonal components. Accordingly,
\[
\mathbf{t}_{ij}
=
\begin{bmatrix}
(\mathbf{s}_j^{\top}\mathbf{u}_{ij} - \mathbf{s}_i^{\top}\mathbf{u}_{ij})\mathbf{u}_{ij}
\\[2ex]
(\mathbf{s}_j - \mathbf{s}_i) + ((\mathbf{s}_i^{\top} - \mathbf{s}_j^{\top})\mathbf{u}_{ij})
\mathbf{u}_{ij}
\end{bmatrix}.
\]

This decomposition reveals two complementary sources of local incompatibility. The first component characterizes the mismatch of vector behavior along the edge direction, while the second component quantifies the mismatch perpendicular to the edge. Together they provide a complete edge-oriented description of neighboring vector consistency.

\subsubsection{Residual Norm as Edge Incompatibility}
To obtain a scalar measure of local inconsistency, we define the $\ell^2$ norm of the sheaf residual by
\begin{equation}\label{Eq. Full sheaf residual norm}
\left\|
\mathbf{t}_{ij}
\right\|
= \sqrt{
\left\| \mathbf{t}_{ij}^{\parallel}\right\|^2 + \left\|  \mathbf{t}_{ij}^{\perp}\right\|^2
}.    
\end{equation}

The residual norm satisfies $\left\|
\mathbf{t}_{ij}
\right\|\ge0$, where $\left\|
\mathbf{t}_{ij}
\right\|=0$ indicates perfect local compatibility on edge \(e_{ij}\), and larger values correspond to stronger local incompatibility between the two incident node vectors.  

From another perspective, the proposed sheaf residual is naturally interpreted through the sheaf coboundary operator. Specifically, for any cellular sheaf $\mathcal{F}$ on a finite graph $G$, the sheaf residuals $\mathbf{t}_{ij}$ coincide with the edge-wise components of $\mathbf{C}_{\mathcal{F}}\mathbf{s}$, where $\mathbf{C}_{\mathcal{F}}$ denotes the $0$th sheaf coboundary matrix and $\mathbf{s}$ is the concatenation of the node vectors~\cite{Hu2026Sheaf}. This interpretation is also closely related to the \textit{quadratic potential function}~\cite{zhao2025} and the \textit{sheaf Dirichlet energy}~\cite{bodnar2022neural}, both of which are formulated from the norm (or quadratic form) of $\mathbf{C}_{\mathcal{F}}\mathbf{s}$ to aggregate local inconsistencies into a global energy. While the global residual energy provides an overall measure of incompatibility, our framework also retains the individual edge-wise residuals $\mathbf{t}_{ij}$ together with their associated local scalar quantities, enabling their physical, biological, and social interpretations to be investigated through the distribution and entropy of these local quantities.

Similar edge incompatibility measures can be defined by emphasizing different geometric characteristics of the node vectors. For example, if the directional discrepancy between the perpendicular components is not of primary interest, one may instead compare only their magnitudes and define

\begin{equation}\label{Eq. Sheaf residue energy}
r_{ij} = \sqrt{ \left\|\mathbf{t}_{ij}^{\parallel}\right\|^2 +  \left(\left\|\mathbf{s}_j- (\mathbf{s}_j^{\top}\mathbf{u}_{ij}) \mathbf{u}_{ij}\right\| - \left\|\mathbf{s}_i - (\mathbf{s}_i^{\top}\mathbf{u}_{ij})
\mathbf{u}_{ij}\right\|\right)^2}    
\end{equation}

By the reverse triangle inequality, $r_{ij}\leq \|\mathbf{t}_{ij}\|$. Compared with the full sheaf residual norm~\eqref{Eq. Full sheaf residual norm}, $r_{ij}$ suppresses directional variations within the perpendicular subspace and retains only magnitude differences. In particular, by~\eqref{Eq. F(vi to eij) = F(vj to eij)}, if the local configurations associated with $\mathbf{s}_i$ and $\mathbf{s}_j$ differ only by a rigid motion (translation or rotation) on the edge $e_{ij}$, then $r_{ij}=0$. Consequently, $r_{ij}$ reduces the influence of orientation-dependent fluctuations and provides a more stable edge-wise descriptor for subsequent statistical analysis. In the applications presented in this paper, we adopt $r_{ij}$ as the edge incompatibility measure.

\subsubsection{Interpretation of the Sheaf Residual}

The proposed sheaf residual quantifies the deviation from local sheaf compatibility by expressing neighboring node vectors in a common edge-associated coordinate system. Unlike a conventional vector difference, it measures incompatibility after transporting local vector information to a common edge stalk.

From the perspective of cellular sheaf theory, the collection of all residuals
\begin{equation}\label{Eq. Collection of all residuals}
\mathcal{T}
=
\left\{
\mathbf{t}_{ij}
\mid
e_{ij}\in E
\right\}    
\end{equation}
characterizes how much the vector field $(G,\mathcal{S})$ departs from a globally compatible assignment of node vectors (i.e., a \textit{global section})~\cite{hansen2019toward,hansen2019learning,robinson2020assignments,moustafa2026finite}. Consequently, the residual field provides a local description of incompatibility throughout the graph.

However, characterizing the organization of a vector field on a graph requires considering not only the magnitude of individual residuals, but also how these residuals are distributed across the graph. To capture this global organization, the next section converts edge residuals into residual energies and quantifies their distribution using an entropy-based formulation.

\subsection{Sheaf Residual Energy Field}
The sheaf residual introduced in the previous section quantifies local incompatibility on each graph edge. However, residuals themselves only describe local discrepancies and cannot directly characterize the global organization of a vector field on a geometric graph. To bridge local incompatibility and global information-theoretic analysis, we further introduce a residual energy field defined over the graph edges.

\subsubsection{Residual Energy}

For each edge $e_{ij}\in E$, let $r_{ij}$ denote the measurement of $e_{ij}$ induced by the sheaf residual norm~\eqref{Eq. Sheaf residue energy}. Rather than using the residual magnitude directly, we define the corresponding residual energy by
\[
\varepsilon_{ij}=r_{ij}^{2}.
\]
The quadratic form guarantees non-negativity, allowing every edge to contribute positively to the overall incompatibility of the graph. At the same time, it naturally emphasizes large local incompatibilities while suppressing small fluctuations, making the resulting energy more sensitive to structurally significant deviations. Quadratic energy functionals are also widely used in graph signal processing, elasticity theory, and sheaf Laplacian formulations, where they provide a natural measure of local inconsistency~\cite{martinez2014smoothed,sloboda2026sheaf,osting2014minimal,ji2026does,chen2024manifold}. Consequently, the proposed residual energy admits both a natural physical interpretation and a mathematically well-established formulation for quantifying edge-wise incompatibility.

\subsubsection{Sheaf Residual Energy Field}
The residual energies associated with the edges naturally define a scalar field on the graph,
\begin{equation}\label{Eq. Sheaf residual energy field}
\mathcal{E}
=
\left\{
\varepsilon_{ij}
\mid
e_{ij}\in E
\right\},    
\end{equation}
which we refer to as the \emph{sheaf residual energy field}. This field may equivalently be regarded as a non-negative function
\[
\varepsilon:
E
\rightarrow
\mathbb{R}_{\geq0},
\]
assigning a residual energy to each edge of the graph. Unlike the collection of residuals $\mathcal{T}$ (see~\eqref{Eq. Collection of all residuals}), which encodes vector-valued local incompatibilities, the sheaf residual energy field provides a scalar representation of their magnitudes. Consequently, the vector field on the graph is transformed into an edge-wise energy landscape, revealing the spatial distribution of local incompatibilities.

\subsubsection{Global Residual Energy}

Given an arbitrary cellular sheaf on a graph $G=(V,E)$, the corresponding collection of sheaf residuals $\mathcal{T}$ naturally defines the quadratic functional
\[
\Phi(\mathcal{T})
=
\sum_{e_{ij}\in E}
\|\mathbf{t}_{ij}\|^{2}.
\]
Each term $\|\mathbf{t}_{ij}\|^{2}$ represents the local contribution of the edge $e_{ij}$ to the global residual energy. This formulation is closely related to existing quadratic energy functionals in sheaf theory, including the quadratic potential function, sheaf Dirichlet energy, and sheaf norm; see~\cite{zhao2025,Hu2026Sheaf,bodnar2022neural}.

Based on the sheaf residual energy field defined in \eqref{Eq. Sheaf residual energy field}, the \textit{total sheaf residual energy} (or just \textit{total residual energy}) of the vector field $(G,\mathcal{S})$ is defined by
\[
E_{\mathrm{tot}}
=
\sum_{e_{ij}\in E}
\varepsilon_{ij}.
\]
The quantity $E_{\mathrm{tot}}$ measures the overall level of local incompatibility across the graph.

However, the total residual energy alone cannot distinguish different organizational patterns. Two collections of residuals may have the same total energy while exhibiting completely different spatial distributions of local incompatibility. Therefore, the total residual energy alone is insufficient to characterize the organization of the residual field.

\subsubsection{Residual Energy Measure}

The sheaf residual energy field naturally induces a finite measure on the measurable space $(E,2^E)$. Specifically, the \emph{residual energy measure} is defined as the measure $\mu:2^E\longrightarrow\mathbb{R}_{\geq0}$ defined by
\begin{equation}\label{Eq. Residual energy measure}
\mu(A)=\sum_{e_{ij}\in A}\varepsilon_{ij}, \qquad A\subseteq E.
\end{equation}
Clearly, $\mu$ is finitely additive; that is, $\mu(A\cup B)=\mu(A)+\mu(B)$ whenever $A$ and $B$ are disjoint subsets of $E$. Hence, $\mu$ is a finite measure on the finite measurable space $(E,2^E)$. In particular, for any edge $e_{ij}\in E$,
\[
\mu(\{e_{ij}\})=\varepsilon_{ij},
\]
and
\[
\mu(E)=\sum_{e_{ij}\in E}\varepsilon_{ij}=E_{\mathrm{tot}}.
\]
This measure-theoretic formulation provides a natural bridge between local sheaf residuals and global information-theoretic analysis. The residual energy measure encodes not only the magnitude of local incompatibility but also its distribution across subsets of graph edges. 

\subsection{Information Quantification}
\label{Subsection: Information Quantification}
The residual energy measure introduced in \eqref{Eq. Residual energy measure} provides a measure-theoretic description of how local incompatibilities are distributed over graph edges. To obtain a probabilistic description, we normalize this finite measure to define the \emph{sheaf residual organization}. Based on this probability measure, we further introduce an entropy-based scalar descriptor that quantifies the organization of the residual energy distribution.

\subsubsection{Normalized Residual Energy Measure}

For the finite residual energy measure $\mu:2^E\longrightarrow\mathbb{R}_{\ge0}$ defined in \eqref{Eq. Residual energy measure}, assume that $\mu(E)=E_{\mathrm{tot}}>0$. We normalize $\mu$ to obtain a probability measure on the measurable space $(E,2^E)$. Specifically, for any subset of edges $A\subseteq E$, define
\begin{equation}\label{Eq. probability measure}
P(A)=\frac{\mu(A)}{\mu(E)} =\frac{\mu(A)}{E_{\mathrm{tot}}}.    
\end{equation}
Then $P$ is a well-defined probability measure on $(E,2^E)$. In particular, the probability assigned to each edge $e_{ij}\in E$ is
\[
p_{ij}:=P(\{e_{ij}\})
=\frac{\varepsilon_{ij}}{E_{\mathrm{tot}}},
\]
and these probabilities satisfy
\[
p_{ij}\ge0,
\qquad
\sum_{e_{ij}\in E}p_{ij}=1.
\]
When $E_{\mathrm{tot}}=0$, the residual energy vanishes on every edge; that is, $\varepsilon_{ij}=0$ (equivalently, $r_{ij}=0$) for all $e_{ij}\in E$. In this degenerate case, the vector field is compatible, and the associated probability measure is treated separately.

The normalization removes the influence of the total residual magnitude while preserving the relative organization of residual energy across graph edges. The resulting probability measure $P$, called the \emph{sheaf residual organization}, characterizes the organization of the residual energy distribution and is defined as follows.

\begin{definition}[Sheaf residual organization]
For a vector field \((G,\mathcal{S})\) with finite and nonzero total residual energy, its \textbf{sheaf residual organization} is defined as the normalized residual energy measure
\[
\mathcal{O}(G,\mathcal{S})
=
P.
\]
\end{definition}

This definition separates the magnitude of local incompatibility from its spatial organization. The total residual energy $E_{\mathrm{tot}}$ quantifies the overall amount of incompatibility, whereas $\mathcal{O}(G,\mathcal{S})$ records its normalized distribution over the edge set. Consequently, the sheaf residual organization $\mathcal{O}(G,\mathcal{S})$ is invariant under uniform rescaling of residual energies and depends only on their relative distribution (Proposition \ref{Proposition: Scale invariance of sheaf residual organization}). 

\subsubsection{Entropy-Based Quantification of Residual Organization}

Based on the normalized residual energy measure \(P\) defined in \eqref{Eq. probability measure}, we define the entropy descriptor \cite{bromiley2004shannon} used in SheafIQ, referred to as the \emph{Sheaf Residual Entropy} (SRE), by
\begin{equation}
H=-\sum_{e\in E} p_{e}\log p_{e}.    
\end{equation}
Since the entropy is computed from the normalized residual energy measure, it quantifies the uncertainty of the sheaf residual organization.

To eliminate the influence of graph size, we further define the normalized entropy for graphs with $|E|>1$ by
\[
\widehat{H}=\frac{H}{\log |E|},
\]
where $0\leq \widehat{H}\leq 1$. The assumption $|E|>1$ ensures that $\log |E|>0$, so the normalization is well-defined. The normalized entropy removes the trivial dependence of the entropy on the size of the edge set, allowing the residual organization of vector fields defined on graphs with different numbers of edges to be compared on a common scale.

\begin{proposition}[Bounds and Extremal Cases of SRE]
\label{Proposition: Bounds and Extremal Cases of SRE}
For any vector field \((G,\mathcal{S})\) with nonzero total residual energy and $|E|>1$,
\begin{equation}\label{Eq. Shannon Entropy Inequality} 
0\leq H\leq\log |E|.    
\end{equation}
Equivalently, dividing both sides by $\log |E|$ yields
\[
0\leq\widehat{H}\leq1.
\]
The lower bound is attained if and only if the normalized residual energy measure is concentrated on a single edge; that is, there exists an edge $e^\ast\in E$ such that
\begin{equation}\label{Eq. Lower bound condition}
p_{e^\ast}=1,
\qquad
p_e=0
\quad
\text{for all } e\neq e^\ast.    
\end{equation}
On the other hand, the upper bound is attained if and only if the residual energy is uniformly distributed over all edges, that is,
\[
p_e
=
\frac{1}{|E|}
\quad
\text{for all } e\in E.
\]
\end{proposition}
\noindent\textit{Proof.}
Let $U$ be the uniform distribution on the edge set $E$. The Kullback--Leibler divergence between the probability distributions $P$ and $U$ leads to the following inequality:
\begin{equation}\label{Eq. KL Divergence}
\operatorname{D}_{\operatorname{KL}}(P\Vert U) = \sum_{e\in E}
p_e\log \left( \frac{p_e}{1/n}  \right)
\geq 0.    
\end{equation}
Equivalently, we have
\begin{equation*}
\sum_{e\in E}
p_e\log \left( p_e  \right) +  \log |E| =
\sum_{e\in E}
p_e\log \left( p_e  \right) + \sum_{e\in E}
p_e\log \left( n  \right) \geq 0     
\end{equation*}
Since \(P\) is a probability measure on the finite edge set \(E\), the Shannon entropy satisfies
\[
0
\leq
-
\sum_{e\in E}
p_e\log p_e
\leq
\log |E|.
\]
Since $-p_e \log p_e$ is always nonnegative, the lower bound is achieved if and only if $-p_e \log p_e = 0$ for all $e \in E$, which is equivalent to the condition \eqref{Eq. Lower bound condition}. 

On the other hand, the equality of the Kullback--Leibler divergence in \eqref{Eq. KL Divergence} holds if and only if $p_e = 1/n$ for all $e \in E$, which is equivalent to $\sum_{e\in E} p_e\log \left( p_e  \right) +  \log |E| = 0$, i.e., the upper bound of \eqref{Eq. Shannon Entropy Inequality} is attained. 
\hfill$\square$

\begin{proposition}[Scale invariance of sheaf residual organization]\label{Proposition: Scale invariance of sheaf residual organization}
Let $\mathcal{E}=\{\varepsilon_e\}_{e\in E}$ be the sheaf residual energy field of a vector field $(G,\mathcal{S})$ on a geometric graph $G = (V,E)$. For any constant $c>0,$ define a scaled energy field
\[
\mathcal{E}^{(c)}
=
\{c\varepsilon_e\}_{e\in E}.
\]
Then the normalized residual energy measure remains unchanged:
\[
P^{(c)}
=
P.
\]
Consequently,
\[
H^{(c)}
=
H,
\qquad
\widehat{H}^{(c)}
=
\widehat{H}.
\]
\end{proposition}

\noindent\textit{Proof.}
For any edge \(e\in E\), the normalized probability under the scaled energy field is
\[
p_e^{(c)}
=
\frac{c\varepsilon_e}
{\sum_{e'\in E}c\varepsilon_{e'}}
=
\frac{\varepsilon_e}
{\sum_{e'\in E}\varepsilon_{e'}}
=
p_e.
\]
Therefore, the normalized residual energy measure is unchanged, and the entropy computed from this measure is also unchanged.
\hfill$\square$

\begin{definition}[SheafIQ operator]
Let \(\mathfrak{G}\) denote the class of vector fields on graphs with finite and nonzero residual energy and with more than one edge. The \textbf{SheafIQ operator} is defined as the assignment
\[
\mathcal{H}_{\mathrm{SheafIQ}}
:
\mathfrak{G}
\longrightarrow
[0,1],
\]
with
\[
\mathcal{H}_{\mathrm{SheafIQ}}(G,\mathcal{S})
=
\widehat{H}.
\]
Thus, \(\mathcal{H}_{\mathrm{SheafIQ}}\) maps each vector field to a normalized information-theoretic descriptor of its sheaf residual organization.
\end{definition}

\subsubsection{Local Incompatibility and Global Organization}

Propositions \ref{Proposition: Bounds and Extremal Cases of SRE} and \ref{Proposition: Scale invariance of sheaf residual organization} clarify the information-theoretic role of the proposed entropy. In particular, the entropy characterizes the relative organization of local sheaf incompatibilities over graph edges rather than their absolute magnitude.

Intuitively, large entropy indicates that residual energy is broadly distributed across many edges, whereas small entropy indicates that it is concentrated on only a few edges. When the residual energy is concentrated on a small number of edges, the probability distribution \(p_{ij}\) becomes highly localized, resulting in a relatively small entropy. This indicates that local incompatibilities are confined to a limited portion of the graph. Conversely, when the residual energy is distributed more uniformly across the edge set, the entropy increases, indicating that local incompatibilities are dispersed more broadly throughout the graph.

Consequently, the entropy characterizes how local incompatibilities are globally organized and provides a quantitative descriptor of the distribution of residual energy in graph-based vector fields. Its scale invariance further ensures that vector fields with proportional residual energies have the same entropy, provided that the relative distribution of residual energy across the graph remains unchanged.

This interpretation distinguishes SheafIQ from conventional graph entropy measures that are derived solely from graph topology, degree distributions, or graph spectra. In contrast, SheafIQ explicitly incorporates the compatibility of vector-valued data under the local coordinate transformations encoded by the sheaf structure.

Overall, SheafIQ consists of two complementary components. The sheaf representation models local vector compatibility through edge-induced coordinate transformations, whereas the entropy summarizes the global organization of the resulting residual energy distribution. In this way, SheafIQ provides an information-theoretic characterization of vector fields on geometric graphs by connecting local sheaf geometry with the global organization of local incompatibilities.

\clearpage

\section{Results}
\subsection{SheafIQ Reveals Mutation-Induced Dynamical Reorganization}
\label{Subsection: SheafIQ Reveals Mutation-Induced Dynamical Reorganization}

\begin{figure}[htbp]
	\captionsetup{font=footnotesize, justification=justified,  singlelinecheck=false, skip=3pt}
	\includegraphics[width=1.0\textwidth]{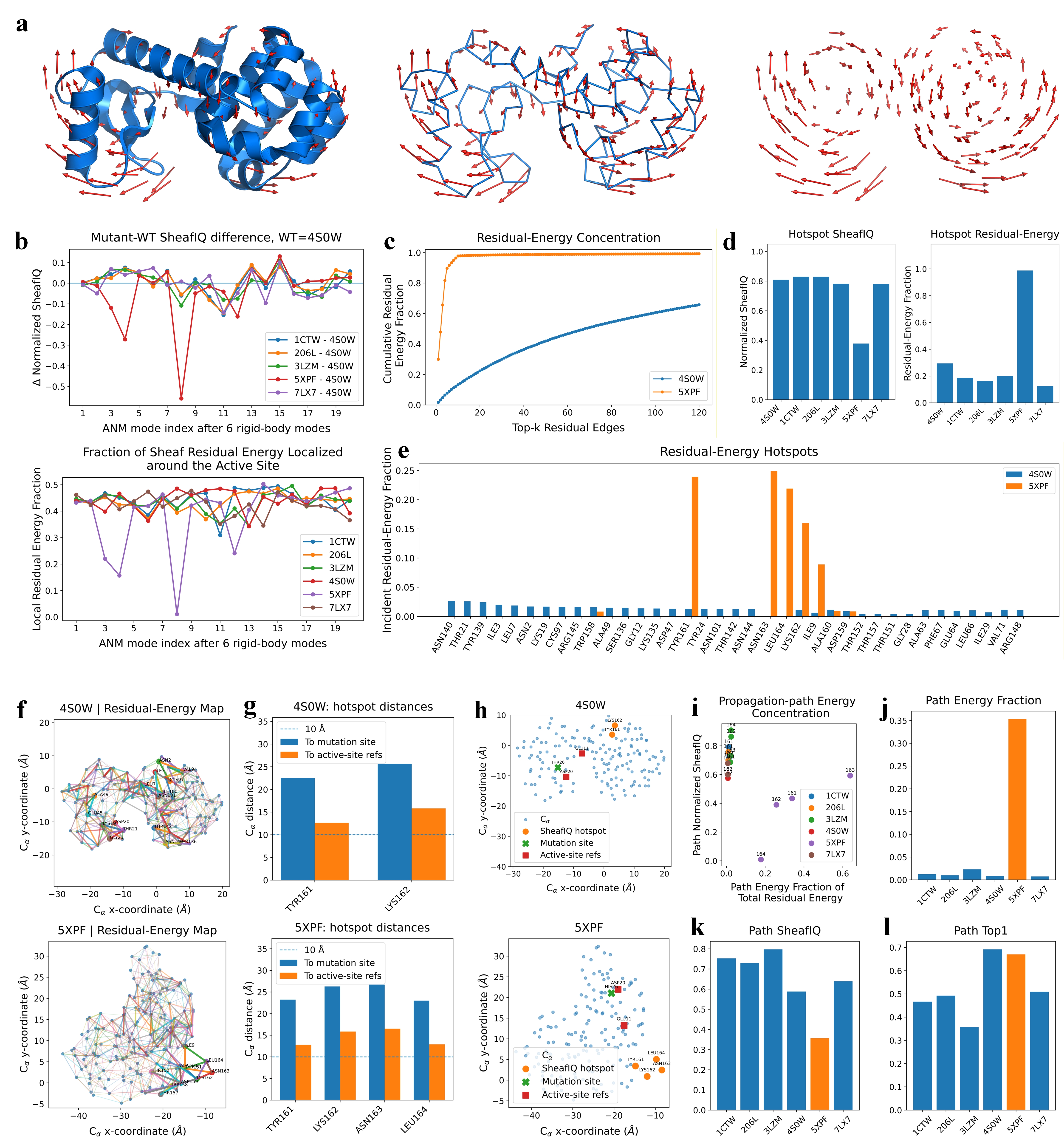}
	\caption{\textbf{SheafIQ identifies mutation-induced residual energy hotspots and long-range communication reorganization in T4 lysozyme.} (a) Construction of the vector field from ANM residue displacement vectors on the ${\rm C}_\alpha$ contact network.
(b) Comparison of normalized SheafIQ spectra and local residual energy fractions across T4 lysozyme mutants. The local residual energy fraction quantifies the proportion of total sheaf residual energy localized within 10$\AA$ of catalytic residues (E11 and D20).
(c) Cumulative residual energy distributions over the highest-energy contact edges for WT (4S0W) and mutant 5XPF at mode~8.
(d) Hotspot SheafIQ and hotspot residual energy fractions.
(e) residual energy hotspot residues.
(f) residual energy maps projected onto the residue contact network.
(g) ${\rm C}_\alpha$ distances between hotspot residues, the mutation site, and known active-site residues.
(h) Spatial localization of hotspot residues relative to the mutation site and catalytic region.
(i) Relationship between path SheafIQ and path-energy concentration.
(j--l) Comparison of path-energy fraction, path SheafIQ, and Top1 path-energy fraction across representative structures.}
	
	\label{mainfig:t4_lysozyme_1} 
\end{figure}

Protein function depends not only on its static structure but also on the spatial organization of collective residue motions. Mutations can alter these motions and reorganize long-range communication without necessarily inducing substantial global structural changes \cite{smith2016allosteric, berendsen2000collective, liu2017dynamical}. To evaluate its ability to capture protein dynamics, SheafIQ was applied to characterize mutation-induced changes and identify the structural regions underlying these changes. For each protein structure, we constructed a ${\rm C}_\alpha$ contact network and calculated its anisotropic network model (ANM) normal modes \cite{kelly1995anisotropic, eyal2006anisotropic, atilgan2001anisotropy}. The displacement vector of each residue in a given mode was treated as a vector field on a geometric graph, allowing SheafIQ to quantify how residual incompatibility energy is distributed across residue contacts. We first used T4 lysozyme \cite{matthews1995studies, bouvignies2011solution} as a detailed case study. We then examined TEM-1 $\beta$-lactamase \cite{sideraki2001secondary, jacquier2013capturing} and dihydrofolate reductase (DHFR) \cite{radkiewicz2000protein, schnell2004structure} to evaluate whether the observed patterns generalize across proteins with diverse structural architectures, functional regions, and mutational or conformational properties.

Figure~\ref{mainfig:t4_lysozyme_1}(a) illustrates the construction of a vector field from an ANM normal mode of T4 lysozyme. We computed SheafIQ for the first 20 non-rigid modes and compared each mutant with the wild-type structure 4S0W. As shown in Figure~\ref{mainfig:t4_lysozyme_1}(b), most mutants exhibited only moderate mode-dependent differences, whereas 5XPF displayed a pronounced decrease in SheafIQ, particularly in mode~8. The local residual energy fraction around the catalytic region exhibited a corresponding increase, identifying mode~8 as the primary focus of the subsequent analyses.

To investigate the structural basis of this pronounced change, we examined the spatial organization of residual energy. As shown in Figure~\ref{mainfig:t4_lysozyme_1}(c), the residual energy in 5XPF became concentrated on a small number of contact edges, whereas that of the wild-type structure remained broadly distributed. We defined residual-energy hotspots as the top-5 residues ranked by aggregated edge residual energy, together with their neighboring residues within 10$\AA$. Consistent with this observation, the hotspot region in 5XPF exhibited substantially lower SheafIQ but a much higher residual energy fraction (Figure~\ref{mainfig:t4_lysozyme_1}(d)). The corresponding hotspot residues and residual energy maps further localized this concentration to a specific structural region (Figures~\ref{mainfig:t4_lysozyme_1}(e,f)). Notably, these hotspot residues were located closer to the functional active-site region than to the mutation site (Figures \ref{mainfig:t4_lysozyme_1}(g,h)), suggesting that mutation-induced dynamical reorganization extends beyond the local mutation environment toward functionally relevant regions. 

To investigate how the mutation-associated residual reorganization extends from the mutation site to the remote hotspot, we further analyzed residual-energy-associated propagation pathways between the mutation residue (T26) and hotspot residues (161–164). For each hotspot residue, we identified the top five residual-energy-weighted shortest paths, where edges with larger residual energy were preferentially selected as low-cost connections. These pathways therefore characterize structural routes along which sheaf residual energy is organized between the mutation site and remote hotspot regions. Compared with the wild-type structure, 5XPF exhibited substantially higher pathway energy fractions, indicating that a larger proportion of total residual energy was concentrated within a limited number of mutation–hotspot communication routes (Figures~\ref{mainfig:t4_lysozyme_1}(i--j)). Meanwhile, the decreased path SheafIQ and increased Top-1 path energy fraction revealed that residual energy within these pathways became further localized to a few dominant edges, suggesting the emergence of a concentrated residual-energy communication backbone after mutation  (Figures~\ref{mainfig:t4_lysozyme_1}(k--l)). Further analyses showed that the mutation-associated residual reorganization was supported by concentrated hotspot residues, bottleneck edges, and compact communication backbones. Additional T4 lysozyme mutants exhibited distinct SheafIQ changes and hotspot organizations, indicating mutation-specific reorganization patterns. Detailed analyses of communication backbones and additional protein systems are provided in Supplementary Section~\ref{Additional T4 Lysozyme Mutants}.

To further evaluate the generality of these findings, we extended the analysis to TEM-1 $\beta$-lactamase and DHFR. Across these systems, SheafIQ consistently revealed mutation- or variant-specific dynamical hotspots and long-range communication reorganization associated with known functional regions. Detailed analyses are provided in Supplementary Section~\ref{Generalization to TEM-1}--~\ref{Statistical Validation of SheafIQ Hotspots in DHFR}.

Together, these results demonstrate that SheafIQ captures mutation-induced reorganization of dynamical communication rather than merely changes in residual magnitude. By identifying residual-energy hotspots and their associated communication backbones, SheafIQ provides an interpretable view of how local mutations reshape long-range residue interactions.

\subsection{Characterizing Protein Conformational Organization with SheafIQ}

Protein conformational changes involve coordinated motions of many residues rather than independent local displacements \cite{lumry1954conformation, koshland1998conformational}. We therefore investigated whether SheafIQ can characterize the organization of these coordinated structural transitions and identify functionally important regions. To this end, we analyzed representative open-to-closed conformational transitions by representing residue displacements as vector fields and characterizing their residual organization using SheafIQ.

\begin{figure}[htbp]
	\captionsetup{font=footnotesize, justification=justified,  singlelinecheck=false, skip=3pt}
	\includegraphics[width=1.0\textwidth]{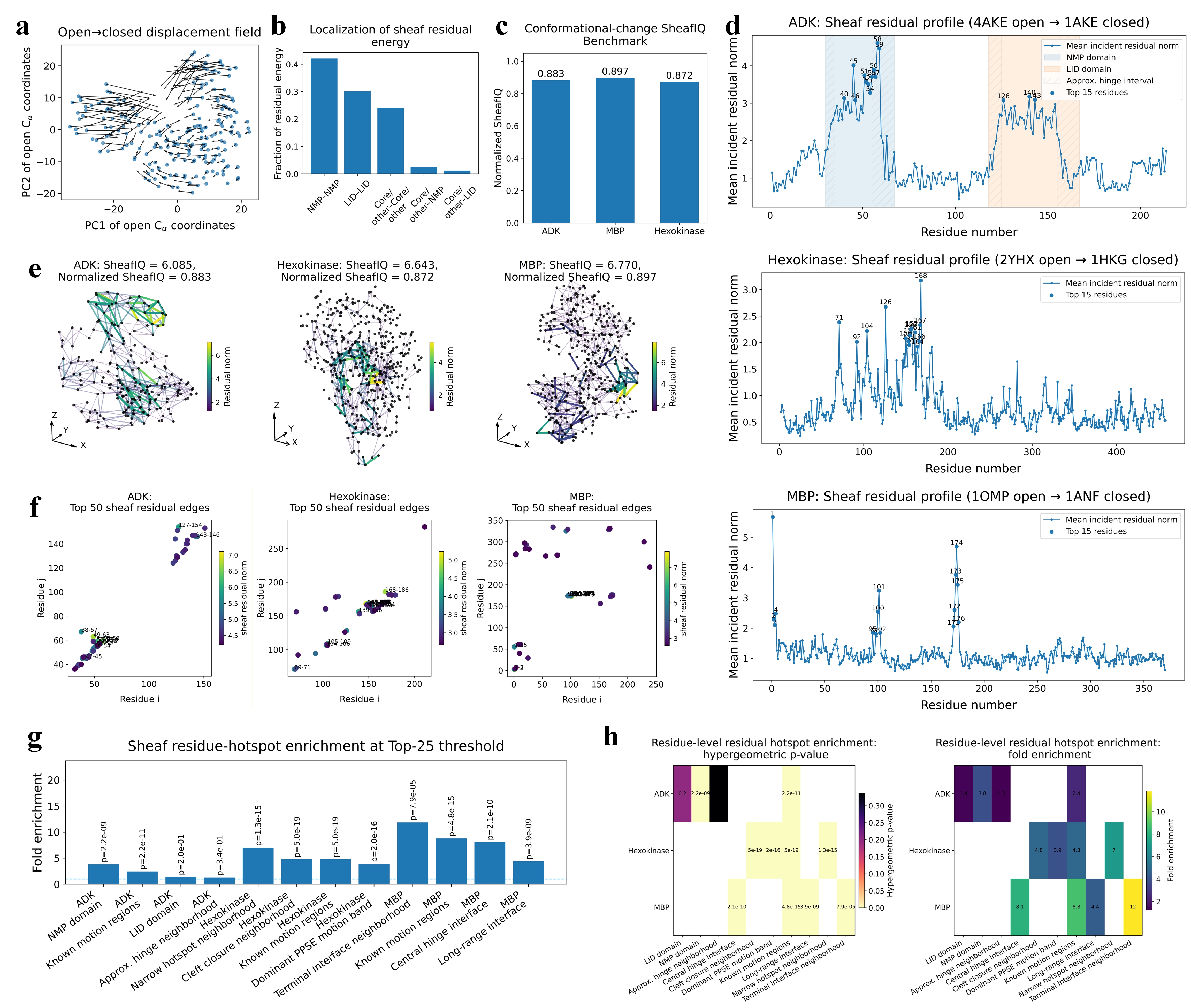}
	\caption{
\textbf{SheafIQ analysis of protein conformational changes.}
(a) Open-to-closed displacement field of adenylate kinase (ADK), where each arrow represents an aligned residue displacement vector.
(b) Distribution of sheaf residual energy across the structural regions of ADK.
(c) Normalized Information Quantity measured by SheafIQ for three representative conformational transitions: ADK, hexokinase, and maltose-binding protein (MBP).
(d) Residue-level sheaf residual profiles for the three proteins. Highlighted regions indicate annotated functional domains or approximate hinge regions, and labeled points denote the top residual residues.
(e) Three-dimensional visualization of high-residual interactions mapped onto the protein structures. Edge colors represent sheaf residual norms, illustrating the spatial organization of conformational changes.
(f) Sequence-coordinate maps of the top 50 sheaf residual edges, highlighting residue pairs with the strongest local incompatibilities.
(g) Fold enrichment of the Top-25 sheaf residual hotspots in experimentally characterized conformationally important regions. Corresponding hypergeometric $p$-values are shown above each bar.
(h) Statistical validation of residue-level hotspot enrichment across the three proteins. The left panel shows hypergeometric $p$-values, and the right panel shows the corresponding fold enrichments for the annotated functional regions.}	
\label{mainfig:molecular_conformation} 
\end{figure}

We first applied SheafIQ to the open-to-closed transition of adenylate kinase (ADK) \cite{arora2007large, olsson2010overlap}. After aligning the two conformations, residue displacements were represented as a vector field on the residue interaction network. As shown in Fig.~\ref{mainfig:molecular_conformation}(a), the transition exhibited coordinated yet spatially heterogeneous residue motions. The resulting sheaf residual energy was predominantly concentrated within the NMP and LID domains (Fig.~\ref{mainfig:molecular_conformation}(b)), indicating that the conformational transition is organized around these principal mobile domains. This heterogeneous residual organization provides the foundation for the subsequent information-theoretic characterization by SheafIQ.

We next evaluated SheafIQ on three representative conformational transitions, including ADK, maltose-binding protein (MBP) \cite{tang2007open, shen2026two}, and hexokinase \cite{bennett1978glucose}. All three proteins exhibited consistently high normalized SheafIQ values (Fig.~\ref{mainfig:molecular_conformation}(c)), indicating that coordinated conformational transitions possess rich residual organization despite substantial differences in protein size and architecture. Residue-level profiles further revealed localized residual peaks (Fig.~\ref{mainfig:molecular_conformation}(d)), while three-dimensional residual maps and the corresponding top residual edges demonstrated that these high-residual interactions formed spatially organized clusters rather than being randomly distributed (Fig.~\ref{mainfig:molecular_conformation}(e,f)).

To evaluate the biological relevance of these residual organizations, we examined whether the identified hotspots corresponded to experimentally reported conformationally important regions. Significant enrichments were consistently observed across all three proteins (Fig.~\ref{mainfig:molecular_conformation}(g,h)). In ADK, hotspots preferentially localized to the NMP domain and known motion regions, whereas in hexokinase and MBP, they showed strong enrichments in annotated cleft-closing, hinge, and domain-interface regions. These results demonstrate that SheafIQ accurately identifies structurally and functionally important regions involved in conformational transitions.

To further verify that these enrichments arise from coordinated conformational organization rather than trivial geometric properties, we compared the observed conformational transitions with two negative controls generated by either randomly permuting residue displacement vectors or randomizing their directions while preserving displacement magnitudes (Fig.~\ref{supfig:molecular_conformation_2}).

Collectively, these results demonstrate that SheafIQ provides a unified framework for characterizing protein conformational changes from the perspective of residual organization. Beyond quantifying the global organization of coordinated residue motions, SheafIQ localizes biologically meaningful conformational hotspots and identifies interaction patterns underlying functional structural transitions. These findings establish SheafIQ as an effective framework for analyzing conformational organization across diverse protein systems.

\subsection{Characterizing Brain Functional Organization with SheafIQ}

\begin{figure}[htbp]
	\captionsetup{font=footnotesize, justification=justified,  singlelinecheck=false, skip=3pt}
	\includegraphics[width=1.0\textwidth]{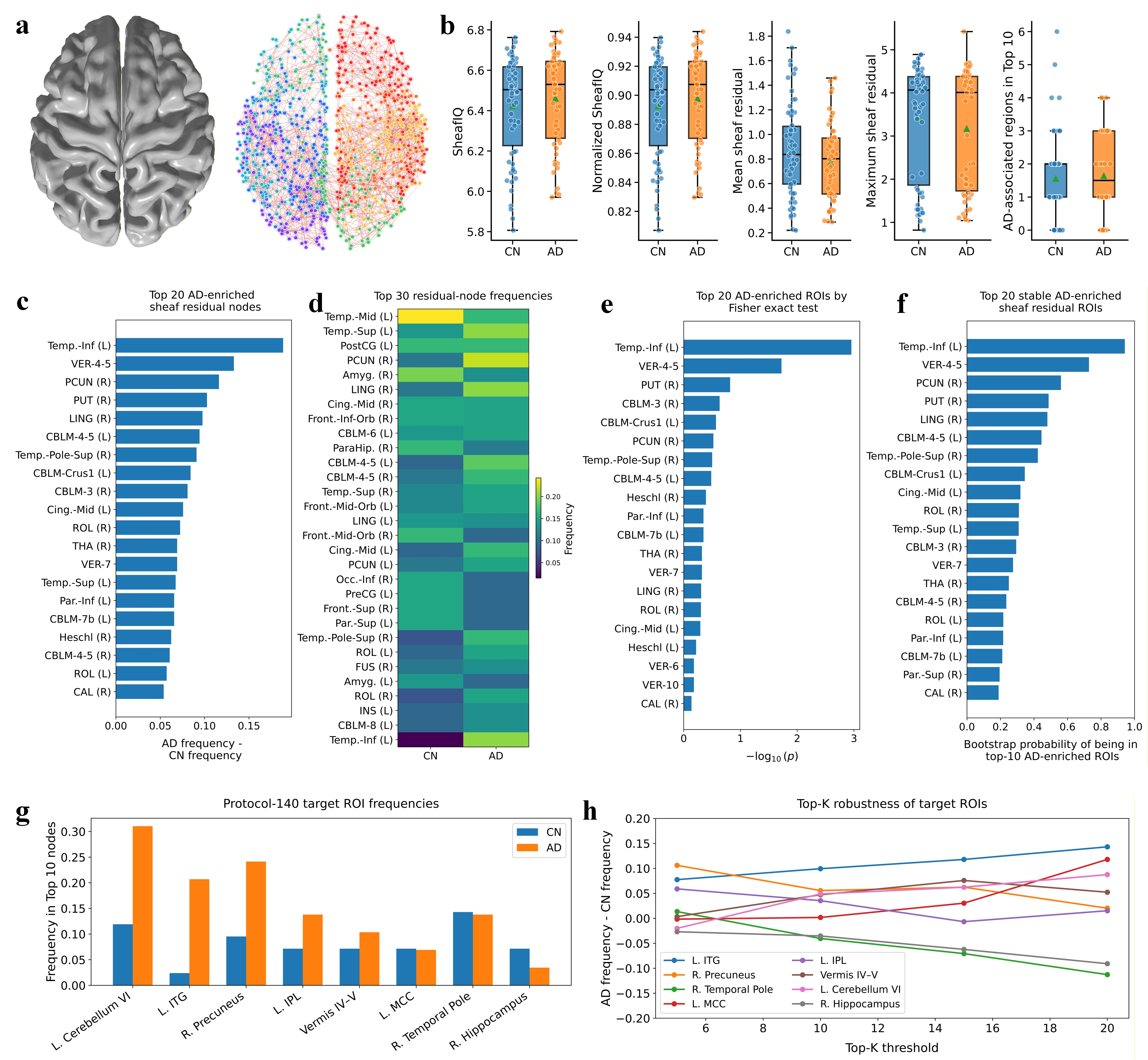}
	\caption{
\textbf{SheafIQ analysis of functional brain networks in Alzheimer's disease.}
(a) Illustration of the brain functional network constructed from the AAL atlas. Brain regions are represented as graph nodes, and functional connections are represented as graph edges.
(b) Comparison of global SheafIQ-related measures between cognitively normal (CN) subjects and subjects with Alzheimer's disease (AD), including SheafIQ, normalized SheafIQ, mean sheaf residual, maximum sheaf residual, and the number of predefined AD-associated regions among the Top-10 sheaf residual nodes.
(c) Top 20 brain regions exhibiting the largest increases in hotspot frequency in subjects with AD relative to CN subjects.
(d) Occurrence frequencies of the Top-30 sheaf residual nodes across CN and AD subjects.
(e) Top 20 AD-associated brain regions ranked by Fisher's exact test. Bar lengths represent $-\log_{10}(p)$.
(f) Top 20 stable AD-associated brain regions identified by subject-level bootstrap analysis. Bar lengths indicate the bootstrap probability of appearing among the Top-10 AD-associated regions.
(g) Frequencies of representative AD-associated brain regions after restricting the analysis to subjects whose imaging data were acquired under the standardized protocol (140 time points and 116 ROIs).
(h) Robustness of representative AD-associated brain regions across different Top-$K$ thresholds (Top-5, Top-10, Top-15, and Top-20). The vertical axis represents the difference in hotspot frequency between subjects with AD and CN subjects.
}	
	\label{mainfig:brain_network} 
\end{figure}

Brain functional disorders are often associated with subtle reorganizations of large-scale functional interactions rather than isolated abnormalities in individual brain regions \cite{stam2009graph, stam2007small}. We therefore investigated whether SheafIQ can characterize disease-related changes in brain functional organization and identify localized residual hotspots associated with Alzheimer's disease (AD). Resting-state fMRI data from cognitively normal (CN) subjects and patients with AD were obtained from the Alzheimer's Disease Neuroimaging Initiative (ADNI) database \cite{petersen2010alzheimer}. Functional connectivity networks were then constructed using the Automated Anatomical Labeling (AAL) atlas \cite{tzourio2002automated}, in which brain regions were represented as graph nodes and functional connections as graph edges (Fig.~\ref{mainfig:brain_network}(a)).

SheafIQ was first applied to characterize the global residual organization of brain functional networks. As shown in Fig.~\ref{mainfig:brain_network}(b), the overall SheafIQ, normalized SheafIQ, and residual statistics were comparable between CN subjects and patients with AD. However, subjects with AD tended to exhibit a greater number of literature-supported AD-associated regions among the Top-10 sheaf residual nodes. These predefined regions included medial temporal memory-related regions, as well as posterior default-mode and temporal regions that have been frequently implicated in Alzheimer's disease \cite{greicius2004default, rao2022hippocampus, galton2001differing}. These findings suggest that disease-related alterations are reflected more strongly in the spatial organization of residual hotspots than in global network-level measures.

We next examined the spatial organization of sheaf residual hotspots. Several brain regions, including the left inferior temporal gyrus, bilateral cerebellar regions, the right precuneus, the right putamen, and the bilateral lingual gyri, were more frequently identified among the Top-10 sheaf residual nodes in subjects with AD than in CN subjects (Fig.~\ref{mainfig:brain_network}(c)). Moreover, these regions were consistently identified across independent subjects with AD while remaining substantially less frequent in CN subjects (Fig.~\ref{mainfig:brain_network}(d)), indicating a reproducible disease-associated hotspot organization. The biological relevance of these hotspot regions was further evaluated through statistical enrichment analysis. Fisher's exact test \cite{upton1992fisher} identified several significantly enriched AD-associated regions, with the left inferior temporal gyrus exhibiting the strongest statistical significance (Fig.~\ref{mainfig:brain_network}(e)). Subject-level bootstrap analysis \cite{wu1986jackknife, hesterberg2011bootstrap} further demonstrated that these major hotspot regions were consistently identified across repeated resampling, confirming their robustness (Fig.~\ref{mainfig:brain_network}(f)).

Finally, we evaluated the robustness of the identified hotspot organization under different analysis settings. The major AD-associated hotspot regions remained consistently identified after restricting the analysis to subjects acquired under a standardized imaging protocol (Fig.~\ref{mainfig:brain_network}(g)). Furthermore, similar enrichment patterns were observed across Top-5, Top-10, Top-15, and Top-20 hotspot selections (Fig.~\ref{mainfig:brain_network}(h)), indicating that the identified hotspot organization is robust to variations in both acquisition conditions and analysis parameters.

Taken together, these results demonstrate that SheafIQ reveals reproducible disease-associated patterns of functional reorganization in brain networks. Rather than relying solely on global network-level measures, SheafIQ identifies localized residual hotspots that are statistically significant, reproducible across subjects, and robust to variations in acquisition protocols and analysis parameters.

\subsection{Characterizing Urban Traffic Organization with SheafIQ}

\begin{figure}[htbp]
	\captionsetup{font=footnotesize, justification=justified,  singlelinecheck=false, skip=3pt}
	\includegraphics[width=1.0\textwidth]{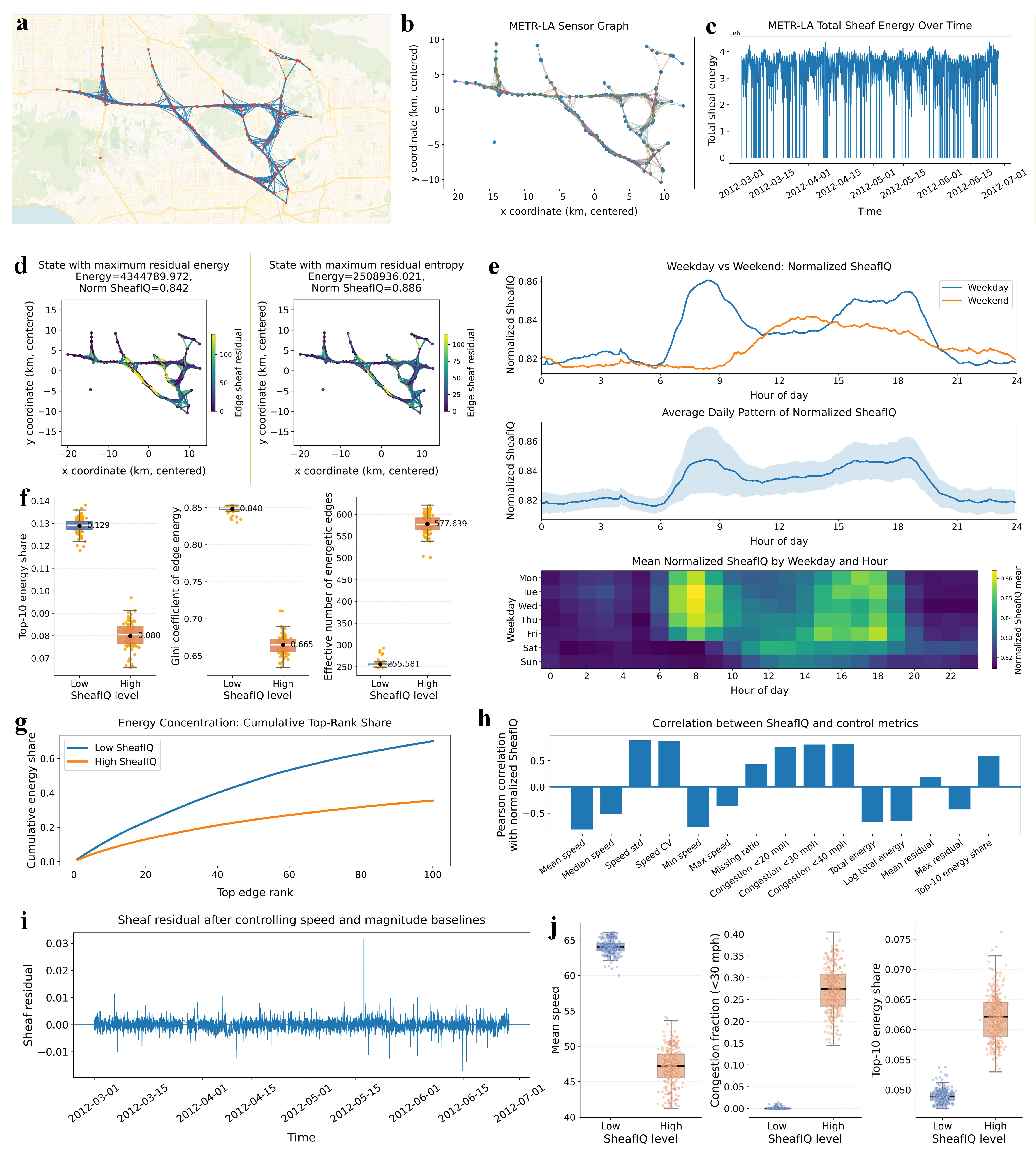}
	\caption{\textbf{Characterization of urban traffic organization using SheafIQ on the METR-LA dataset.}
(a) Geographic locations of traffic sensors and road connections.
(b) Graph representation of the METR-LA traffic network.
(c) Temporal evolution of total sheaf residual energy.
(d) Representative traffic states with the maximum total sheaf residual energy (left) and maximum normalized SheafIQ (right).
(e) Temporal patterns of normalized SheafIQ. The top panel compares weekdays and weekends, the middle panel shows the average daily profile, and the bottom panel presents a weekday--hour heatmap.
(f) Comparison of residual energy organization between low- and high-SheafIQ states using the Top-10 energy share, Gini coefficient of edge residuals, and effective number of energetic edges.
(g) Cumulative residual energy distributions ranked by edge residual magnitude for low- and high-SheafIQ states.
(h) Pearson correlations between normalized SheafIQ and representative traffic metrics and residual-based statistics.
(i) Residual component of normalized SheafIQ after regressing out the traffic speed and residual magnitude.
(j) Comparison of representative traffic metrics between low- and high-SheafIQ states.}
	\label{mainfig:METR-LA_1} 
\end{figure}

Urban traffic systems exhibit highly dynamic spatiotemporal interactions, where congestion emerges from coordinated changes across many road segments rather than isolated local events. We therefore investigated whether SheafIQ can characterize the organization of traffic flow and identify representative network states from large-scale sensor observations. The METR-LA dataset \cite{LiYS018} was represented as a graph, in which traffic sensors correspond to graph nodes and road connections define graph edges (Fig.~\ref{mainfig:METR-LA_1}(a,b)). At each time point, traffic speeds were projected onto locally estimated road directions to construct a vector field on the underlying network. SheafIQ was then computed to characterize the organization of traffic residual interactions.

We first examined the temporal evolution of traffic organization. As shown in Fig.~\ref{mainfig:METR-LA_1}(c) and Fig.~\ref{supfig:METR-LA_2}(a), both total sheaf residual energy and normalized SheafIQ varied continuously over time, indicating persistent changes in network-wide traffic organization. Averaging across multiple days further revealed highly reproducible daily patterns (Fig.~\ref{mainfig:METR-LA_1}(e); Fig.~\ref{supfig:METR-LA_2}(c)). SheafIQ increased rapidly during weekday morning rush hours, remained elevated throughout the day, and exhibited a second peak during the evening commute, whereas weekend profiles were substantially smoother. These findings demonstrate that SheafIQ naturally captures the regular temporal organization of urban traffic dynamics.

We next investigated how traffic organization differs across representative network states. Visualizing the states with the maximum sheaf residual energy and the maximum normalized SheafIQ (Fig.~\ref{mainfig:METR-LA_1}(d); Fig.~\ref{supfig:METR-LA_2}(b)) showed that the maximum-energy state contained widespread high-residual road segments, whereas the maximum-SheafIQ state exhibited a more spatially distributed residual pattern despite having a lower overall residual magnitude. This distinction became even more evident when comparing low- and high-SheafIQ states. High-SheafIQ states displayed substantially lower Top-10 energy shares and Gini coefficients \cite{dorfman1979formula}, together with a much larger effective number of energetic edges (Fig.~\ref{mainfig:METR-LA_1}(f)). Correspondingly, residual energy accumulated much more gradually across ranked edges (Fig.~\ref{mainfig:METR-LA_1}(g)), and the spatial residual maps showed that the additional residual energy was distributed across many road segments rather than concentrated within a few localized regions (Fig.~\ref{supfig:METR-LA_2}(d--f)). These findings indicate that SheafIQ primarily reflects the spatial organization of traffic interactions rather than simply the magnitude of residual energy.

Finally, we examined whether SheafIQ provides information beyond conventional traffic descriptors. As shown in Fig.~\ref{mainfig:METR-LA_1}(h), normalized SheafIQ was strongly associated with traffic speed, congestion, and several residual-based statistics. The scatter plots further confirmed these relationships while revealing considerable variability among samples with similar values of these conventional metrics (Fig.~\ref{mainfig:METR-LA_1}(j); Fig.~\ref{supfig:METR-LA_2}(g)), suggesting that these variables cannot fully explain the observed SheafIQ patterns. To further account for the effects of traffic speed and residual magnitude, we regressed normalized SheafIQ against these baseline variables and analyzed the remaining residual component. As shown in Fig.~\ref{mainfig:METR-LA_1}(i), a clear temporal structure remained after controlling for these factors. Moreover, the daily profile of SheafIQ differed substantially from those of traffic speed, congestion, total sheaf residual energy, and residual concentration (Fig.~\ref{supfig:METR-LA_2}(h)), demonstrating that SheafIQ captures complementary organizational information beyond existing traffic metrics.

We further evaluated the generality of these findings using the PEMS-BAY traffic dataset \cite{LiYS018}. Applying the same graph construction and SheafIQ analysis pipeline, we observed consistent temporal patterns, spatial residual organization, and relationships with conventional traffic metrics (Fig.~\ref{supfig:PEMS-BAY_1}; Fig.~\ref{supfig:PEMS-BAY_2}). These additional results demonstrate that the organizational patterns captured by SheafIQ are not specific to a single traffic network but generalize across different urban transportation systems.

Overall, experiments on both the METR-LA and PEMS-BAY datasets demonstrate that SheafIQ provides a robust characterization of large-scale traffic organization. Across different traffic networks, SheafIQ consistently captures reproducible temporal dynamics, reveals how residual interactions are spatially organized across road segments, and identifies organizational patterns that cannot be fully explained by conventional traffic metrics such as traffic speed, congestion level, or residual magnitude. These findings establish SheafIQ as a complementary measure for characterizing the structural organization of dynamic traffic systems.

\subsection{Characterizing Power Grid Organization with SheafIQ}
\label{Subsection: Characterizing Power Grid Organization with SheafIQ}

\begin{figure}[htbp]
	\captionsetup{font=footnotesize, justification=justified,  singlelinecheck=false, skip=3pt}
	\includegraphics[width=1.0\textwidth]{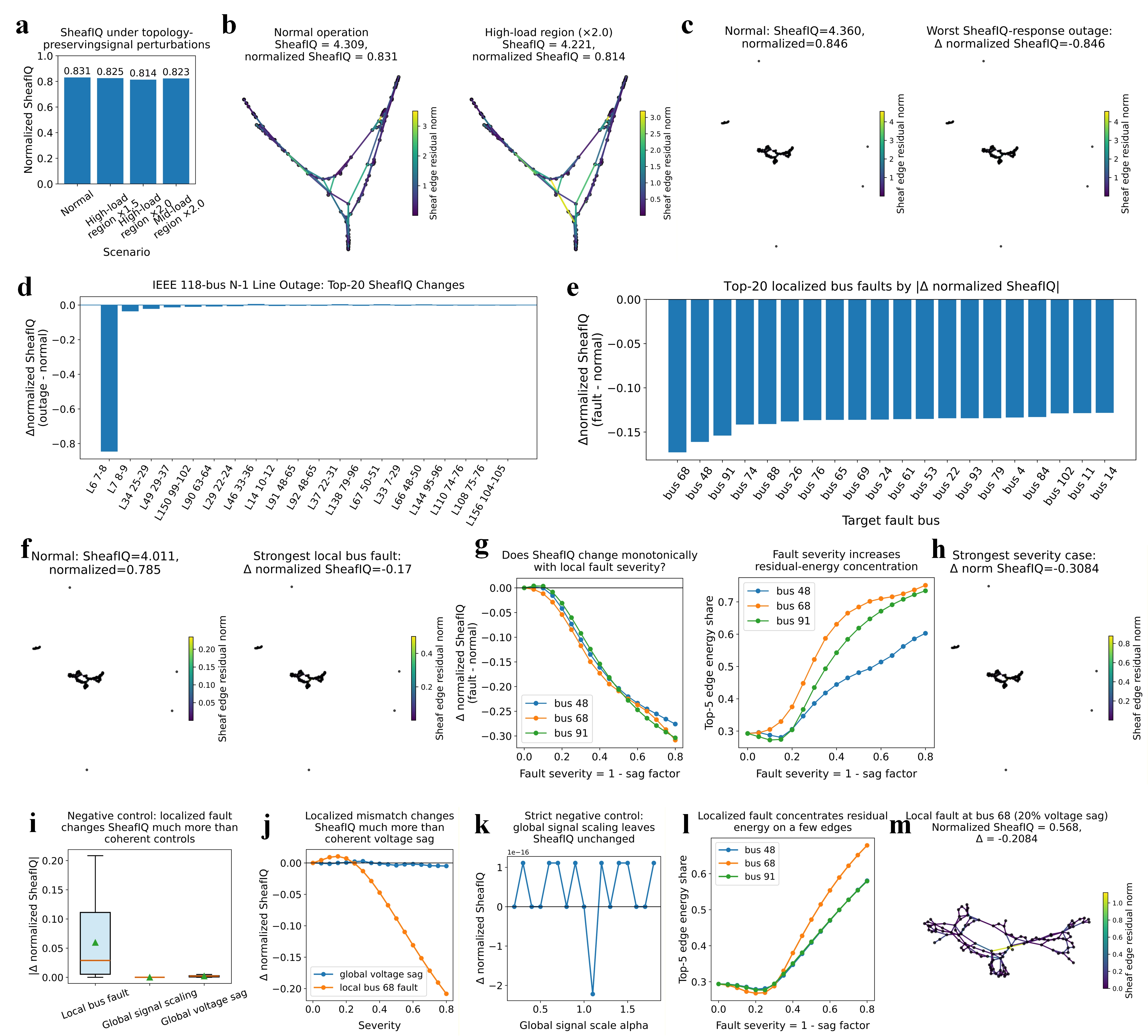}
	\caption{
\textbf{Characterizing electrical-state organization in the IEEE 118-bus power grid using SheafIQ.}
(a) Normalized SheafIQ under topology-preserving regional load perturbations.
(b) Representative edge residual distributions under normal operation and a localized load perturbation.
(c) Representative edge residual distributions under normal operation and the most disruptive transmission-line outage.
(d) Top transmission-line outages ranked by the reduction in normalized SheafIQ.
(e) Top localized bus faults ranked by the reduction in normalized SheafIQ.
(f) Representative edge residual distributions under normal operation and the strongest localized bus fault.
(g) Changes in normalized SheafIQ and the Top-5 edge energy share with increasing fault severity for representative buses.
(h) Representative edge residual distribution under the strongest fault severity.
(i) Comparison of changes in normalized SheafIQ induced by localized bus faults, coherent global signal scaling, and coherent global voltage sag.
(j) Comparison between localized bus faults and coherent global voltage sag under increasing fault severity.
(k) Strict negative control showing normalized SheafIQ under coherent global signal scaling.
(l) Changes in the Top-5 edge energy share with increasing fault severity for representative buses.
(m) Representative edge residual distribution under a localized bus fault.
}
	\label{mainfig:IEEE118} 
\end{figure}

Power grids require maintaining globally coordinated electrical states under continually changing operating conditions and localized disturbances. We therefore investigated whether SheafIQ can characterize the organization of electrical interactions and identify critical operating states in transmission networks. The IEEE 118-bus system \cite{thurner2018pandapower} was represented as a graph, in which buses and transmission lines correspond to graph nodes and edges, respectively. Bus voltage magnitudes and phase angles obtained from AC power-flow analysis \cite{thurner2018pandapower} were treated as node signals (or node vectors), while the network topology was kept fixed unless transmission-line outages were explicitly considered.

We first examined whether SheafIQ responds to changes in electrical states without altering the network topology. Localized load perturbations consistently reduced normalized SheafIQ, and stronger perturbations produced progressively lower values regardless of the perturbed region (Fig.~\ref{mainfig:IEEE118}(a)). The corresponding residual maps (Fig.~\ref{mainfig:IEEE118}(b); Fig.~\ref{supfig:IEEE118_2}(a)) showed that these perturbations reorganized the spatial distribution of voltage residuals while leaving the network topology unchanged. These findings indicate that SheafIQ captures changes in the organization of electrical states rather than graph topology itself.

We next evaluated the response of SheafIQ to transmission-line outages using N-1 contingency analysis \cite{khanabadi2012optimal}. Most line outages caused only minor changes, whereas a small number of critical transmission lines produced substantially larger reductions in normalized SheafIQ (Fig.~\ref{mainfig:IEEE118}(c,d)). Importantly, these responses showed little dependence on pre-outage line loading (Fig.~\ref{supfig:IEEE118_2}(b)), suggesting that SheafIQ reflects the reorganization of network-wide electrical interactions rather than the operating conditions of individual transmission lines.

We further investigated localized bus faults by independently perturbing each bus while preserving the network topology. Different buses produced markedly different reductions in normalized SheafIQ (Fig.~\ref{mainfig:IEEE118}(e)), indicating that only a subset of buses plays a dominant role in maintaining global electrical organization. The strongest-response case exhibited pronounced residual localization around the perturbed bus (Fig.~\ref{mainfig:IEEE118}(f)). Although larger responses were generally associated with higher-degree buses, substantial variability remained among buses with similar degrees (Fig.~\ref{supfig:IEEE118_2}(c)), suggesting that SheafIQ provides complementary information beyond topology-based importance measures.

We next examined whether SheafIQ continuously tracks fault severity. As fault severity increased, normalized SheafIQ decreased monotonically for all representative buses, whereas residual energy became increasingly concentrated within a small subset of transmission lines (Fig.~\ref{mainfig:IEEE118}(g)). The most severe fault generated highly localized residual activity together with the largest reduction in normalized SheafIQ (Fig.~\ref{mainfig:IEEE118}(h)), and similar monotonic trends were observed across additional buses (Fig.~\ref{supfig:IEEE118_2}(d)).

Finally, we investigated whether these responses arise specifically from localized electrical inconsistencies rather than coherent global signal changes. Compared with localized bus faults, coherent global voltage sag and global signal scaling produced only minor or negligible changes in normalized SheafIQ (Fig.~\ref{mainfig:IEEE118}(i--k)). In contrast, localized faults progressively concentrated residual energy within a small subset of transmission lines and generated pronounced residual hotspots (Fig.~\ref{mainfig:IEEE118}(l,m); Fig.~\ref{supfig:IEEE118_2}(e)). These findings demonstrate that SheafIQ primarily reflects the spatial organization of localized electrical disturbances rather than uniform global variations.

To evaluate the generality of these findings, we further performed the same analyses on the larger IEEE 300-bus system. As shown in Supplementary Figs.~\ref{supfig:IEEE300}, the responses of SheafIQ to localized load perturbations, transmission-line outages, bus faults, fault severity, and the control experiments were highly consistent with those observed for the IEEE 118-bus system. Overall, these findings demonstrate that SheafIQ provides a robust characterization of the organization of electrical interactions, enabling the identification of critical components and the characterization of localized disturbances across power transmission networks.

\section{Discussion and Conclusion}

Graph-based vector fields arise naturally in a wide variety of biological and engineered systems, yet existing analytical frameworks primarily characterize either network topology or graph signals independently. In this work, we introduced SheafIQ, a sheaf-theoretic information quantification framework that explicitly bridges local vector compatibility and global information-theoretic organization. Rather than measuring only the magnitude of local incompatibilities, SheafIQ characterizes how these incompatibilities are spatially organized across a graph through the entropy of the normalized residual energy distribution. This perspective extends information-theoretic network analysis beyond structural topology to the organization of vector fields on graphs.

The experiments across proteins, functional brain networks, urban traffic systems, and power grids consistently demonstrate that this formulation captures complementary organizational information that cannot be explained solely by graph topology, vector magnitudes, or conventional application-specific descriptors. Despite the substantial differences among these systems, they all share a common mathematical structure consisting of a graph, node-associated vectors, and local geometric relationships. By introducing a common edge-associated coordinate system through cellular sheaves, SheafIQ provides a unified mechanism for evaluating local vector compatibility, while the entropy of the resulting residual organization summarizes how local incompatibilities are distributed over the entire network. The consistent observations across these diverse applications suggest that residual organization represents a general characteristic of vector fields on graphs rather than a domain-specific property.

Several limitations should also be acknowledged. First, the current formulation assumes that node vectors can be expressed in a common geometric coordinate system so that meaningful edge-associated reference frames can be constructed. Although this assumption is satisfied by many physical systems, including molecular dynamics, traffic flow, and electrical networks, more general graph-based data may require alternative sheaf constructions or domain-specific restriction maps. Second, SheafIQ characterizes the spatial organization of vector incompatibilities but does not explicitly model temporal evolution, uncertainty, or multiscale interactions. Extending the framework to dynamic sheaf models, probabilistic formulations, or hierarchical graph representations therefore represents promising future directions. It is also of interest to investigate alternative measures of edge incompatibility, beyond the one adopted in this work, together with more general cellular sheaf constructions, such as asymmetric sheaf models. Understanding their geometric interpretations and identifying their suitability for different classes of datasets constitute interesting directions for future research.

Beyond the applications considered here, the proposed framework is potentially applicable to a broad range of systems involving vector-valued interactions on graphs, including mechanical systems, fluid dynamics, robotic swarms, biological tissues, and scientific simulations. More generally, SheafIQ establishes a unified information-theoretic framework for quantifying vector fields on graphs by connecting local geometric compatibility with global organizational complexity, providing a new perspective for studying coordinated organization in complex systems.

\section{Funding}
This research was supported by the National Science and Technology Council, Taiwan [NSTC 115-2115-M-390-001-MY2], the Beijing Natural Science Foundation [1264060], the China Postdoctoral Science Foundation[2026M793419], and the Postdoctoral Fellowship Program of CPSF [GZC20252026].


\clearpage 

\bibliography{SheafIQ} 
\bibliographystyle{sciencemag}




\newpage


\renewcommand{\thefigure}{S\arabic{figure}}
\renewcommand{\thetable}{S\arabic{table}}
\renewcommand{\theequation}{S\arabic{equation}}
\renewcommand{\thepage}{S\arabic{page}}
\setcounter{figure}{0}
\setcounter{table}{0}
\setcounter{equation}{0}
\setcounter{page}{1} 


\begin{center}
\section*{Supplementary Materials for\\ \scititle}

Cong Shen,
Guancen Lin$^{\ast}$,
Chuan-Shen Hu$^{\ast}$\\ 
\small$^\ast$Corresponding author. Email: linguancen@amss.ac.cn, chuanshenhu1@nuk.edu.tw\\
\end{center}



\newpage



\section{Supplementary Results}


\subsection{Generalization and Statistical Validation on Additional Protein Systems}
\label{Generalization and Statistical Validation on Additional Protein Systems}

\subsubsection{Additional T4 Lysozyme Mutants}
\label{Additional T4 Lysozyme Mutants}
To further investigate whether the mutation-induced dynamical reorganization observed in 5XPF is specific to this mutant or represents a broader pattern, we analyzed additional T4 lysozyme mutants using the same SheafIQ-based framework.

For the representative mutant 5XPF, residual-energy-associated propagation pathways were first examined to characterize how the mutation-associated perturbation extends toward remote hotspot regions. The identified pathways were concentrated around a limited number of hotspot residues and bottleneck edges (Figure~\ref{supfig:t4_lysozyme_2}(a--d)). By aggregating the dominant residual-energy-weighted pathways, we obtained compact communication backbones connecting the mutation site and remote hotspot regions (Figure~\ref{supfig:t4_lysozyme_2}(e--g)). These results indicate that mutation-induced dynamical reorganization is mediated through a restricted set of structurally important communication routes rather than being uniformly distributed across the protein structure.

We next extended the analysis to additional T4 lysozyme mutants. As shown in Figure~\ref{supfig:t4_lysozyme_2}(h--j), different mutants exhibited distinct patterns of SheafIQ variation, residual-energy localization, and communication backbone organization. In particular, 1G06 and 1G0J displayed the most pronounced changes, characterized by stronger residual-energy concentration and altered long-range communication patterns. These observations demonstrate that SheafIQ captures mutation-specific reorganization of residue communication and identifies diverse dynamical responses across different T4 lysozyme variants.

\subsubsection{Generalization to TEM-1 $\beta$-lactamase}
\label{Generalization to TEM-1}
To assess the generality of SheafIQ beyond T4 lysozyme, we applied the same analysis to TEM-1 $\beta$-lactamase mutants. As shown in Figures~\ref{supfig:TEM-1}(a--c), different mutants exhibited distinct mode-dependent SheafIQ changes, communication backbone rewiring, and hotspot organizations, with 4RVA exhibiting the largest structural deviation.

We next examined the functional relevance of these hotspots. Figures~\ref{supfig:TEM-1}(d--f) show that hotspot residual energy was consistently enriched within or near the $\Omega$-loop and remained spatially close to the catalytic core, SDN loop, and KTG motif. Across multiple mutants, the largest proportion of hotspot residual energy was concentrated around the $\Omega$-loop.

Together, these findings suggest that SheafIQ consistently localizes mutation-induced dynamical reorganization to biologically important regions of TEM-1.

\subsubsection{Generalization to DHFR Structural Variants}
\label{Generalization to DHFR Structural Variants}
To further examine the generality of SheafIQ across different structural variants, we applied the same analysis to DHFR. Figures~\ref{supfig:DHFR_1}(a--c) show that different variants exhibited distinct SheafIQ changes, communication backbone rewiring, and hotspot organizations, with 3DFR exhibiting the largest structural deviation.

We next examined the biological relevance of these hotspots. As shown in Figures~\ref{supfig:DHFR_1}(d--f), hotspot residues were preferentially enriched within or around the Met20 loop and remained spatially close to the Met20 loop and active-site residues. Different variants exhibited distinct residual energy allocation patterns across functional regions, indicating variant-specific dynamical reorganization.

Together, these findings demonstrate that SheafIQ consistently captures variant-specific dynamical reorganization while preferentially highlighting residues associated with known functional regions.

\subsubsection{Statistical Validation of SheafIQ Hotspots in DHFR}
\label{Statistical Validation of SheafIQ Hotspots in DHFR}
To further examine whether these hotspot organizations are statistically meaningful rather than arising from random fluctuations, we performed randomization tests using randomly sampled residue sets. A representative null distribution is shown in Figure~\ref{supfig:DHFR_2}(a). Across all DHFR variants, hotspot enrichment and hotspot proximity to functional regions were statistically significant (Figures~\ref{supfig:DHFR_2}(b,c)), indicating that the identified hotspots are unlikely to have arisen by chance.

We next performed rank-enrichment analysis to determine whether functionally important residues preferentially receive high SheafIQ scores. Figures~\ref{supfig:DHFR_2}(d--g) show that mutation-related residues together with the Met20 and FG loops were consistently enriched among the top-ranked SheafIQ residues and were recovered substantially earlier than expected under random ranking.

Finally, consensus analysis across all DHFR structures revealed that SheafIQ hotspots were highly reproducible. Figures~\ref{supfig:DHFR_2}(h--j) show that mutation-related residues contributed the largest number of recurrent hotspots, while several residues (e.g., 55, 66, 67, 68, 87, and 128) repeatedly appeared as high-confidence hotspots across different structures. These findings demonstrate that SheafIQ identifies reproducible communication hotspots rather than isolated structure-specific signals.

Taken together, the analyses of TEM-1 and DHFR demonstrate that the communication hotspots identified by SheafIQ are reproducible across diverse protein systems, are significantly enriched in known functional regions, and remain robust under multiple statistical validation analyses.

\subsection{Generalization to the PEMS-BAY Traffic Dataset}
\label{Generalization to the PEMS-BAY Traffic Dataset}
To evaluate the generality of SheafIQ on urban traffic systems, we repeated the same analysis on the PEMS-BAY dataset. Traffic speeds were represented as vector fields on the sensor network, and SheafIQ was computed at each time point to characterize the organization of spatial traffic dynamics. Figures~\ref{supfig:PEMS-BAY_1}(a,b) illustrate the geographic distribution of traffic sensors and the corresponding graph representation.

Consistent with the METR-LA results, both normalized SheafIQ and total sheaf residual energy exhibited substantial temporal variations throughout the observation period (Fig.~\ref{supfig:PEMS-BAY_1}(d)). Representative network states with maximum residual energy, median residual energy, and maximum normalized SheafIQ further demonstrated that high-SheafIQ states were characterized by residual activity distributed across multiple road segments rather than concentrated within a few localized regions (Fig.~\ref{supfig:PEMS-BAY_1}(b,c,e)). Daily averages, weekday--weekend comparisons, and weekday--hour heatmaps all revealed highly reproducible morning and evening commuting peaks, confirming that SheafIQ captures stable temporal organization associated with recurring traffic demand (Fig.~\ref{supfig:PEMS-BAY_1}(f)).

We next examined how the spatial organization of residual interactions differed between low- and high-SheafIQ states. High-SheafIQ states exhibited lower Top-10 energy shares and smaller Gini coefficients together with substantially larger effective numbers of energetic edges (Fig.~\ref{supfig:PEMS-BAY_1}(g)). The corresponding average residual maps, cumulative energy distributions, edge residual histograms, and difference maps consistently showed that residual energy was distributed across a much broader portion of the network rather than concentrated on a limited number of road segments (Fig.~\ref{supfig:PEMS-BAY_1}(h--j)). These findings closely mirror those obtained for METR-LA, indicating that high-SheafIQ states consistently correspond to more spatially distributed traffic organization.

We further performed the same control analyses used for METR-LA. Normalized SheafIQ exhibited strong correlations with conventional traffic metrics, including traffic speed, congestion, and residual-based statistics (Fig.~\ref{supfig:PEMS-BAY_2}(a,e)). After regressing out the effects of speed- and magnitude-related baseline variables, clear temporal structures remained in the residual component (Fig.~\ref{supfig:PEMS-BAY_2}(b)). Moreover, the average daily profile of SheafIQ differed from those of conventional traffic metrics (Fig.~\ref{supfig:PEMS-BAY_2}(c)), while high- and low-SheafIQ states remained clearly distinguishable in terms of traffic characteristics (Fig.~\ref{supfig:PEMS-BAY_2}(d)). These findings indicate that SheafIQ captures organizational information that cannot be explained solely by traffic intensity or residual magnitude.

Overall, the PEMS-BAY experiments closely reproduced the major findings from METR-LA, including consistent temporal organization, spatially distributed residual interactions during high-SheafIQ states, and complementary information beyond conventional traffic metrics. These findings demonstrate that SheafIQ generalizes robustly across different large-scale urban traffic networks.

\subsection{Generalization to the IEEE300 Power Grid}

To further evaluate the generality of SheafIQ, we repeated the localized bus-fault analysis on the larger IEEE300 power grid. As shown in Fig.~\ref{supfig:IEEE300}(a), increasing fault severity consistently reduced normalized SheafIQ while simultaneously increasing the Top-5 residual energy share, indicating progressively stronger concentration of residual energy under more severe local disturbances. The corresponding residual maps (Fig.~\ref{supfig:IEEE300}(b)) further show that the most severe localized fault generated pronounced residual activity around the perturbed region, whereas the remaining network was only weakly affected. These findings closely mirror those obtained for the IEEE118 system.

We next performed the same positive and negative control experiments. As shown in Fig.~\ref{supfig:IEEE300}(c), localized bus faults produced substantially larger reductions in normalized SheafIQ than coherent global voltage sags of comparable severity. Consistently, localized faults progressively concentrated residual energy onto a small subset of transmission lines, whereas normalized SheafIQ remained nearly unchanged under uniform global signal scaling (Fig.~\ref{supfig:IEEE300}(d,e)). Quantitative comparisons further confirmed that localized perturbations induced much larger SheafIQ variations than either negative-control condition (Fig.~\ref{supfig:IEEE300}(f)). The corresponding residual maps (Fig.~\ref{supfig:IEEE300}(g)) likewise showed pronounced residual hotspots under localized faults, while coherent global perturbations largely preserved the overall residual organization.

Taken together, these findings demonstrate that SheafIQ's responses to localized electrical disturbances, fault severity, and control perturbations are highly consistent across the IEEE118 and IEEE300 systems, supporting its robustness and general applicability in characterizing the organization of electrical interactions in power transmission networks.

\begin{figure}[htbp]
	\captionsetup{font=footnotesize, justification=justified,  singlelinecheck=false, skip=3pt}
	\includegraphics[width=1.0\textwidth]{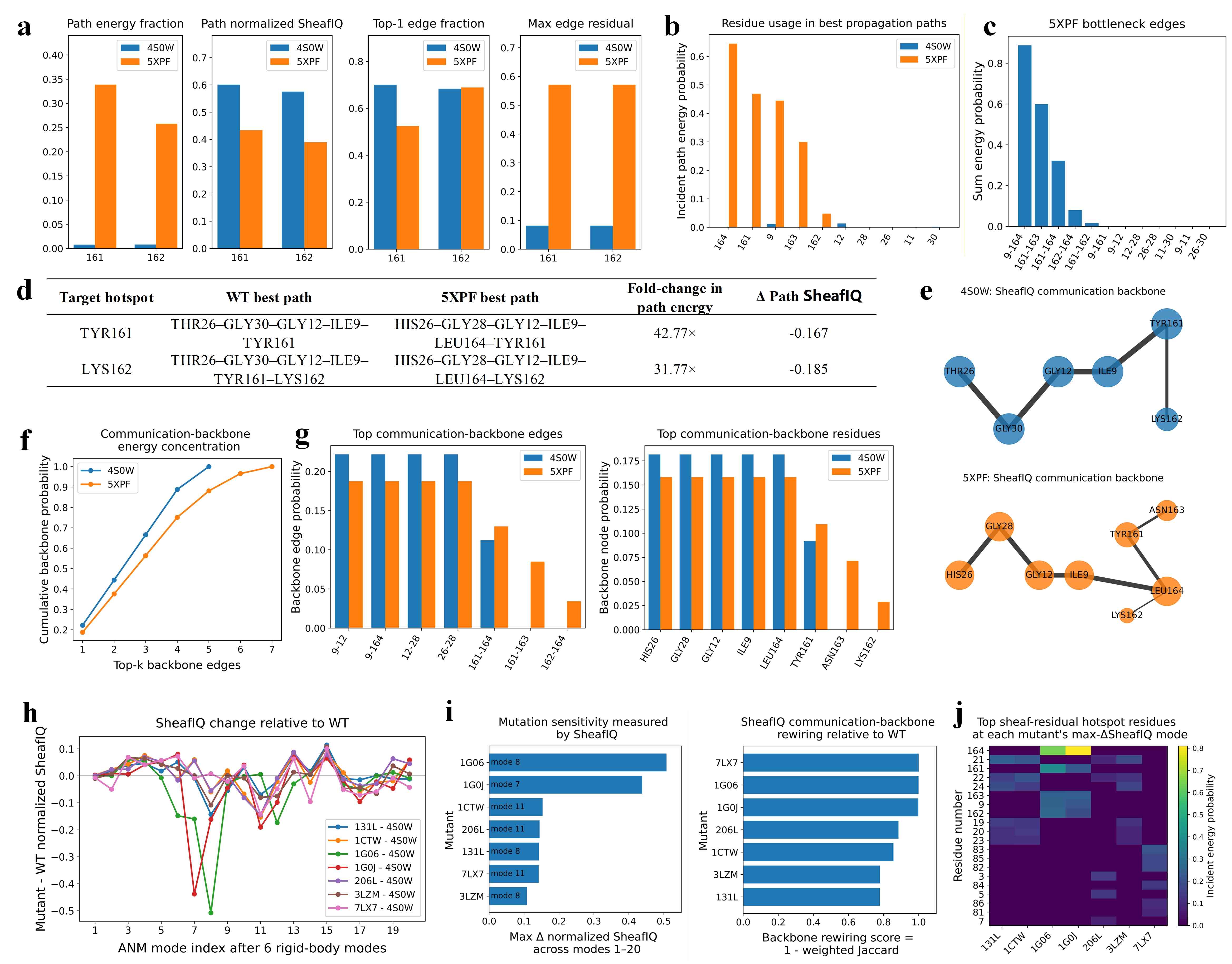}
	\caption{\textbf{Propagation-path analysis and communication backbone organization in T4 lysozyme.}
(a) Comparison of propagation-path energy concentration, path SheafIQ, Top-1 edge fraction, and maximum edge residual between WT and 5XPF.
(b) Residue usage frequencies in the highest-energy propagation pathways.
(c) Dominant bottleneck edges supporting residual energy propagation in 5XPF.
(d) Comparison of the best mutation-to-hotspot propagation pathways between WT and 5XPF.
(e) SheafIQ communication backbones reconstructed from the dominant propagation pathways.
(f) Cumulative residual energy captured by the communication backbone.
(g) Residual energy contributions of the dominant backbone edges and residues.
(h) Mode-dependent SheafIQ changes across additional T4 lysozyme mutants relative to the wild type.
(i) Mutation sensitivity measured by the maximum SheafIQ change and communication backbone rewiring.
(j) Mutation-specific residual energy hotspot residues identified at the maximum $\Delta$SheafIQ mode.
}
	\label{supfig:t4_lysozyme_2} 
\end{figure}

\begin{figure}[htbp]
	\centering
	\captionsetup{font=footnotesize, justification=justified,  singlelinecheck=false, skip=3pt}
	\includegraphics[width=1.0\textwidth]{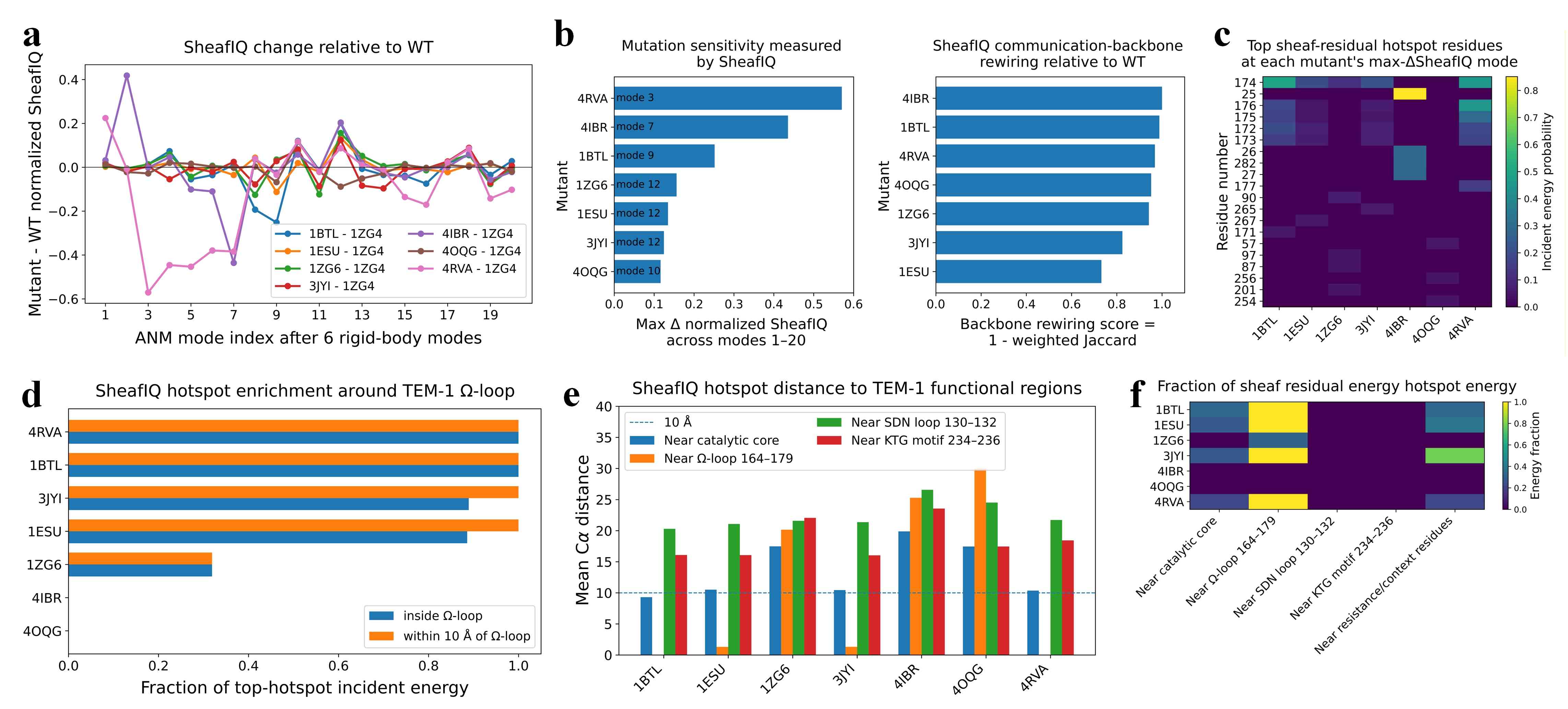}
	\caption{\textbf{SheafIQ identifies mutation-specific structural reorganization and functionally relevant hotspots in TEM-1 $\beta$-lactamase.}
(a) SheafIQ changes of different TEM-1 mutants relative to the wild type (1ZG4) across the first 20 non-rigid ANM modes.
(b) Mutation sensitivity measured by the maximum SheafIQ change (left) and the corresponding communication backbone rewiring score (right).
(c) Heatmap of hotspot residues identified at the mode exhibiting the maximum SheafIQ change for each mutant.
(d) Enrichment of hotspot residual energy within and around the TEM-1 $\Omega$-loop.
(e) Mean ${\rm C}_\alpha$ distances between hotspot residues and representative functional regions, including the catalytic core, $\Omega$-loop, SDN loop, and KTG motif.
(f) Distribution of hotspot residual energy among different functional regions across TEM-1 mutants.
    }

	\label{supfig:TEM-1} 
\end{figure}

\begin{figure}[htbp]
	\captionsetup{font=footnotesize, justification=justified,  singlelinecheck=false, skip=3pt}
	\includegraphics[width=1.0\textwidth]{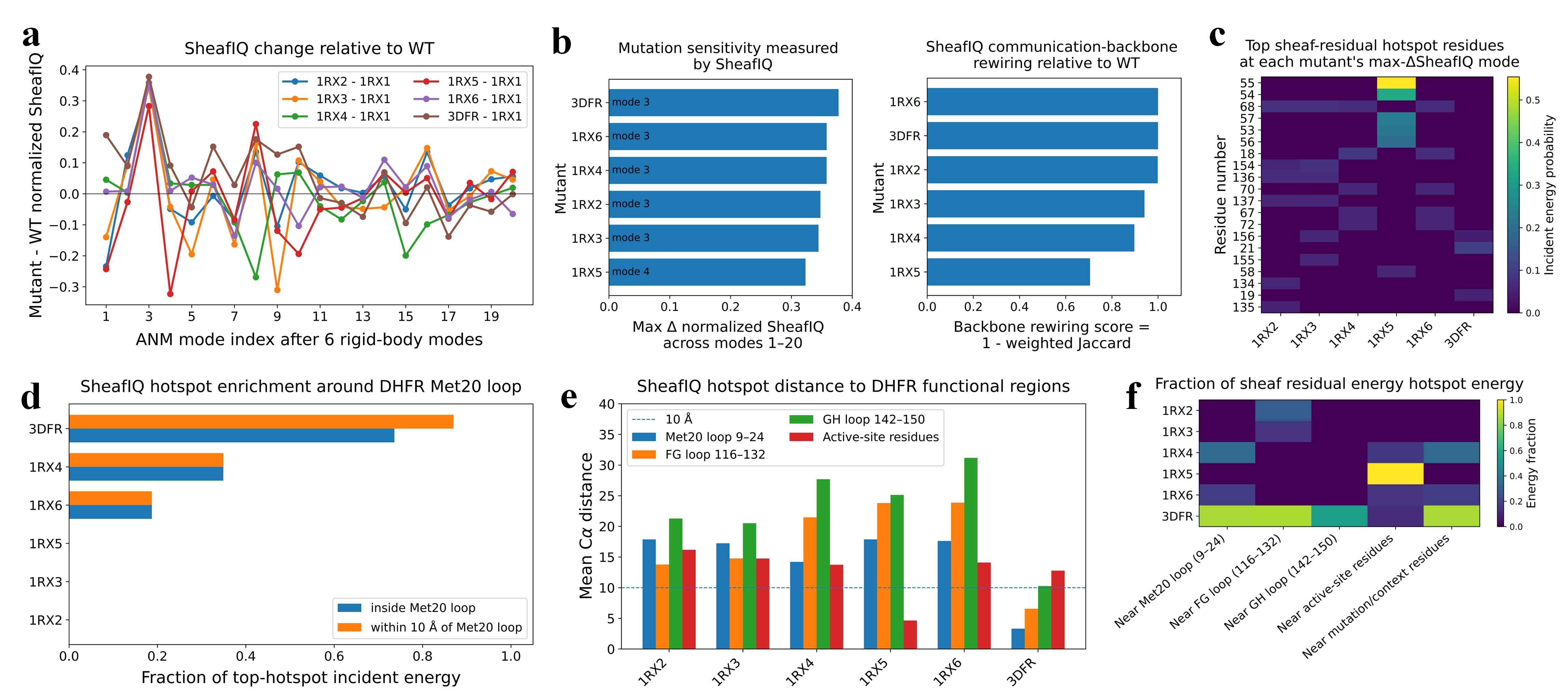}
	\caption{\textbf{SheafIQ identifies variant-specific structural reorganization and functionally relevant hotspots in DHFR.} 
(a) SheafIQ changes of different DHFR variants relative to the wild type across the first 20 non-rigid ANM modes. 
(b) Variant sensitivity quantified by the maximum SheafIQ change (left) and the corresponding communication backbone rewiring score (right). 
(c) Heatmap of hotspot residues identified at the mode exhibiting the maximum $\Delta$SheafIQ for each variant. 
(d) Enrichment of hotspot residual energy within and around the DHFR Met20 loop. 
(e) Mean ${\rm C}_\alpha$ distances between hotspot residues and representative DHFR functional regions, including the Met20 loop, FG loop, GH loop, and active-site residues. 
(f) Distribution of hotspot residual energy among different functional regions across DHFR variants.}

	\label{supfig:DHFR_1} 
\end{figure}

\begin{figure}[htbp]
	\centering
	\captionsetup{font=footnotesize, justification=justified,  singlelinecheck=false, skip=3pt}
	\includegraphics[width=0.9\textwidth]{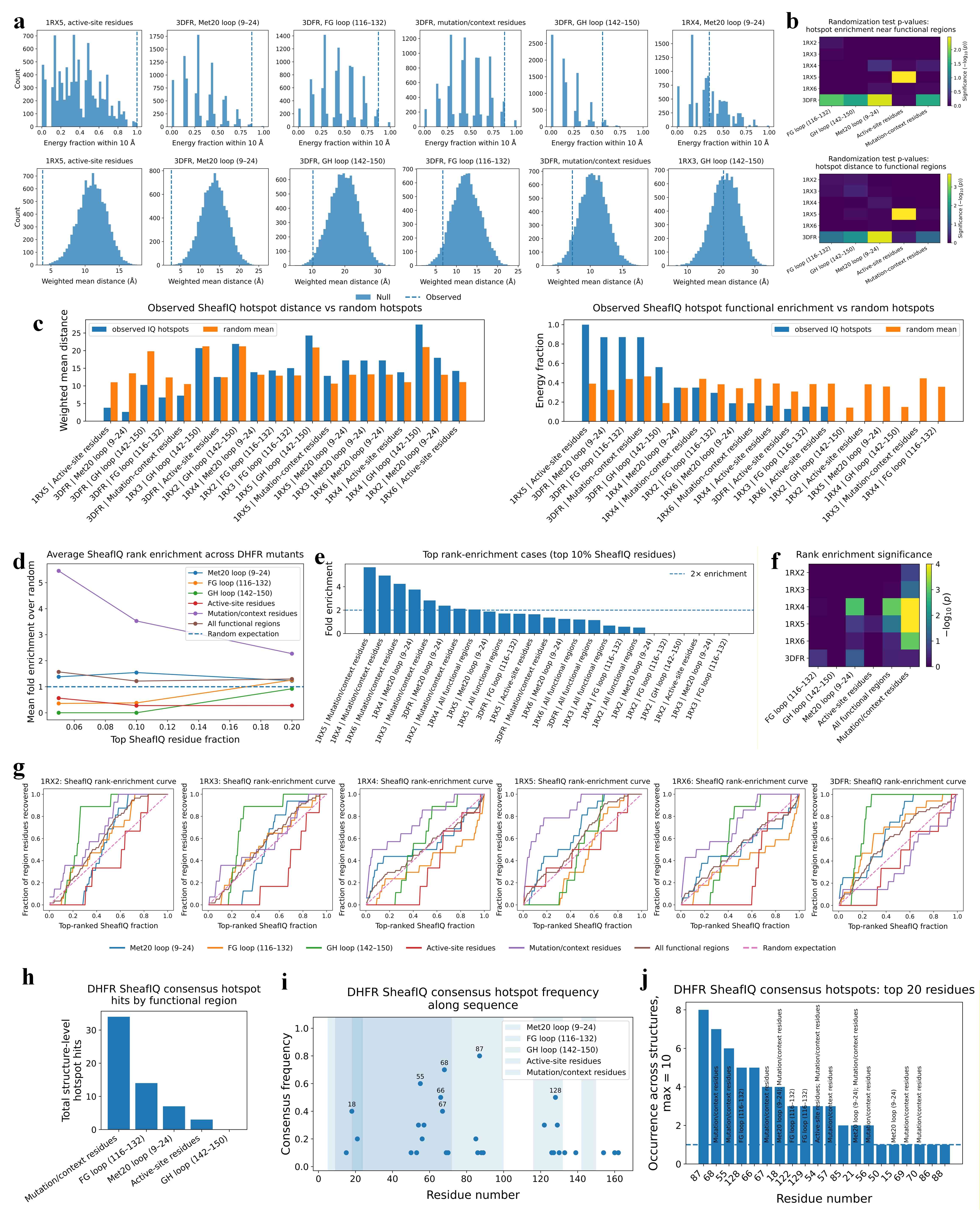}
	\caption{\textbf{Statistical validation and consensus analysis of SheafIQ hotspots across DHFR structures.}
(a) Representative randomization tests comparing the observed hotspot enrichment and hotspot distances with randomly sampled residue sets.
(b) Heatmaps summarizing the statistical significance of hotspot enrichment (top) and hotspot proximity to functional regions (bottom) across DHFR variants.
(c) Comparison of observed SheafIQ hotspots and randomly sampled hotspots in terms of distances to functional regions (left) and hotspot energy fractions (right).
(d) Average rank enrichment of functional regions among the top-ranked SheafIQ residues across DHFR variants.
(e) Representative cases exhibiting the strongest rank enrichment.
(f) Statistical significance of rank enrichment across different functional regions.
(g) Rank-enrichment curves showing the recovery of functional residues as increasingly larger fractions of the top-ranked SheafIQ residues are considered.
(h) Total consensus hotspot occurrences across different DHFR functional regions.
(i) Consensus frequencies of hotspot residues along the DHFR sequence.
(j) Twenty most frequently detected consensus hotspot residues across DHFR structures.}
	\label{supfig:DHFR_2} 
\end{figure}

\begin{figure}[htbp]
	\centering
	\captionsetup{font=footnotesize, justification=justified,  singlelinecheck=false, skip=3pt}
	\includegraphics[width=0.9\textwidth]{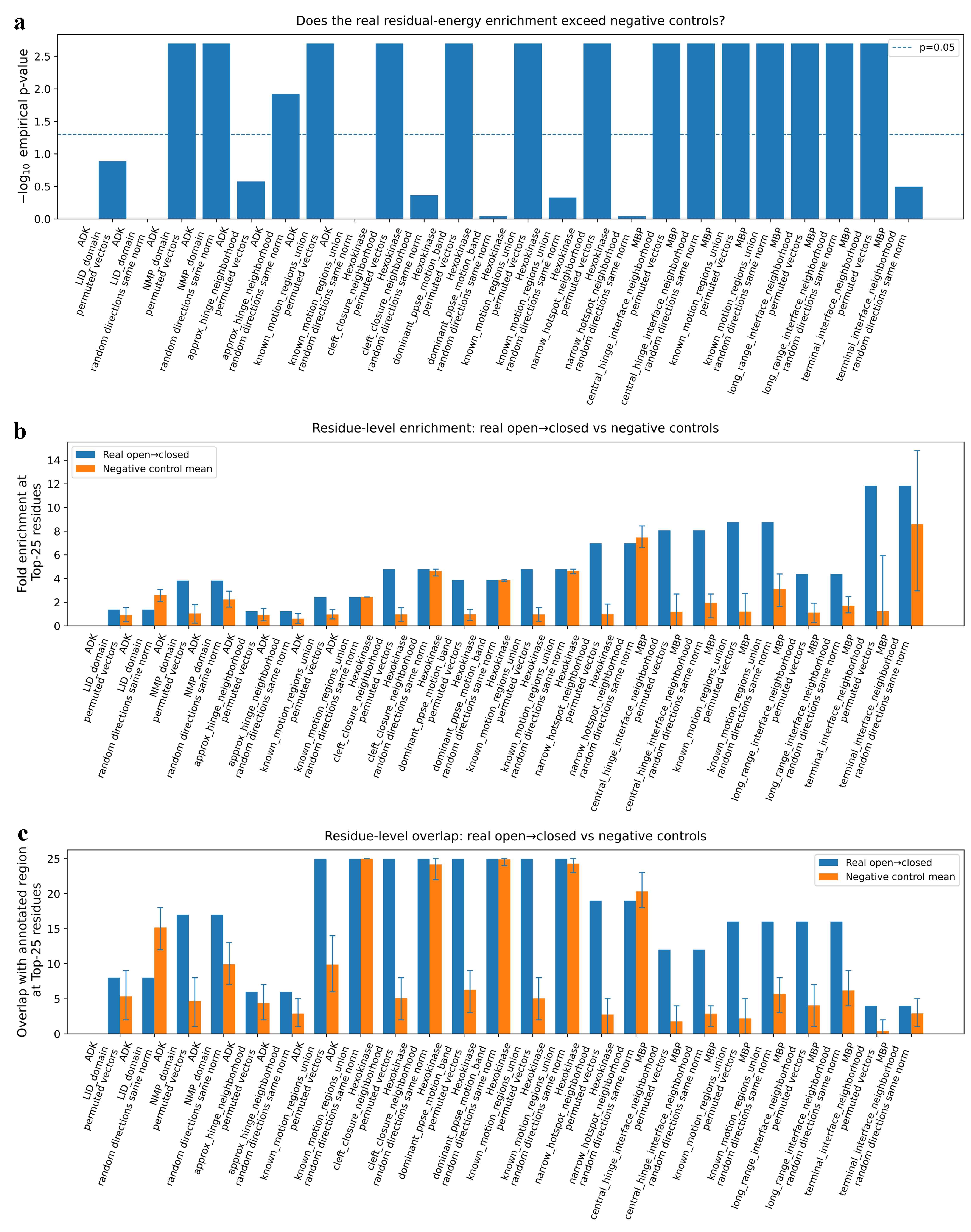}
	\caption{
\textbf{Negative-control validation of SheafIQ hotspot enrichment for protein conformational changes.} (a) Empirical significance of hotspot enrichment for the real conformational transitions relative to the corresponding negative controls. Bars represent $-\log_{10}$ empirical $p$-values, and the dashed line indicates the significance threshold ($p=0.05$).
(b) Comparison of fold enrichment for the Top-25 sheaf residual hotspots between the real open-to-closed conformational transitions and the two negative controls. Error bars denote the 5th--95th percentile range across negative-control realizations.
(c) Comparison of the numbers of Top-25 hotspot residues overlapping annotated conformationally important regions between the real conformational transitions and the two negative controls. Error bars denote the 5th--95th percentile range across negative-control realizations.
}
	\label{supfig:molecular_conformation_2} 
\end{figure}

\begin{figure}[htbp]
	\centering
	\captionsetup{font=footnotesize, justification=justified,  singlelinecheck=false, skip=3pt}
	\includegraphics[width=0.95\textwidth]{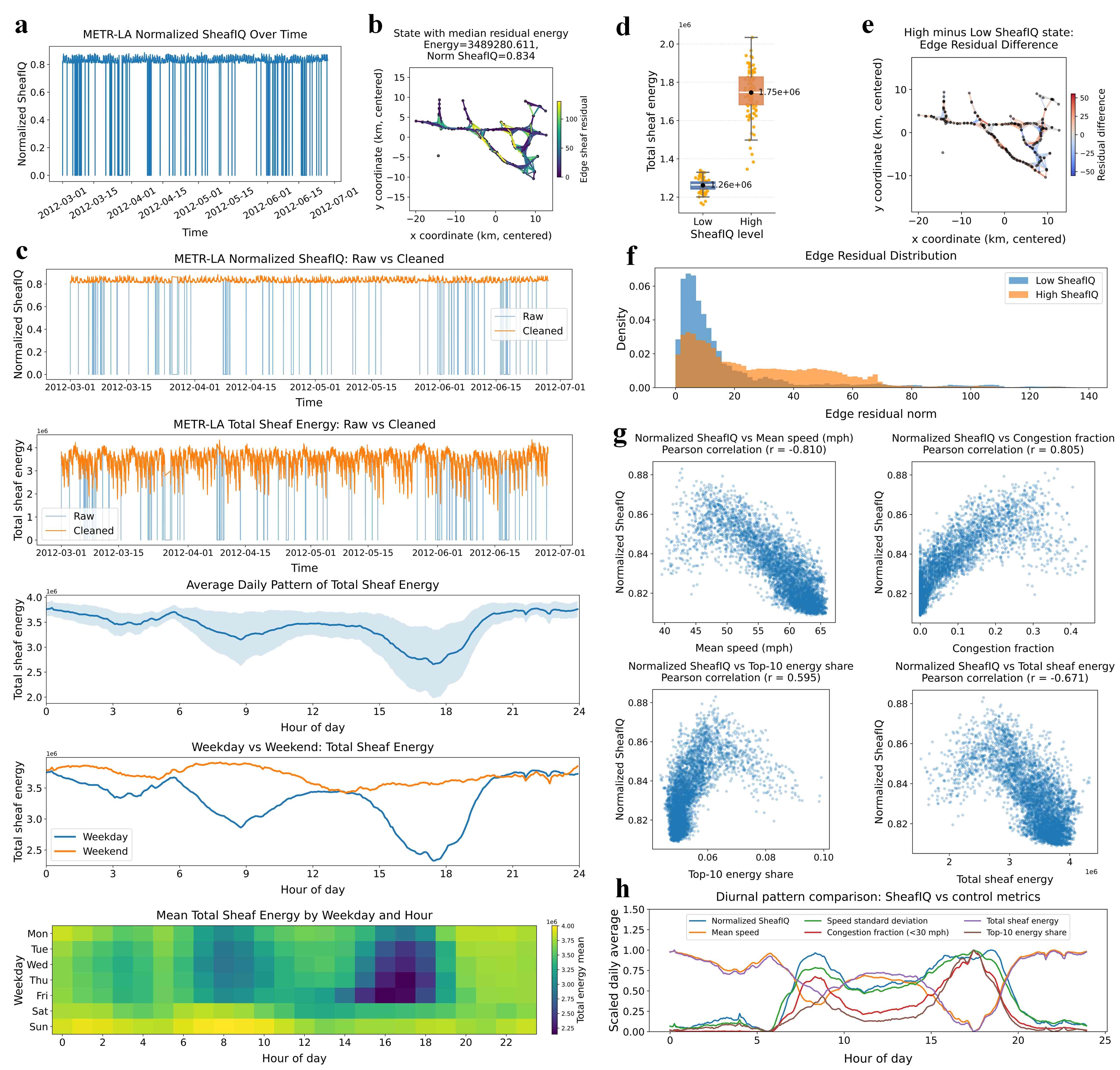}
	\caption{\textbf{Additional analyses of SheafIQ on the METR-LA traffic network.}
(a) Temporal evolution of normalized SheafIQ.
(b) Spatial distribution of edge residuals for a representative high-SheafIQ state.
(c) Comparison of raw and cleaned normalized SheafIQ and total sheaf residual energy, together with average daily profiles, weekday--weekend patterns, and weekday--hour heatmaps of total sheaf residual energy.
(d) Comparison of total sheaf residual energy between low- and high-SheafIQ states.
(e) Difference in edge residuals between representative high- and low-SheafIQ states.
(f) Edge residual distributions for low- and high-SheafIQ states.
(g) Scatter plots showing the relationships between normalized SheafIQ and representative traffic metrics.
(h) Comparison of normalized daily patterns between SheafIQ and conventional traffic descriptors.}
	\label{supfig:METR-LA_2} 
\end{figure}

\begin{figure}[htbp]
	\centering
	\captionsetup{font=footnotesize, justification=justified,  singlelinecheck=false, skip=3pt}
	\includegraphics[width=0.95\textwidth]{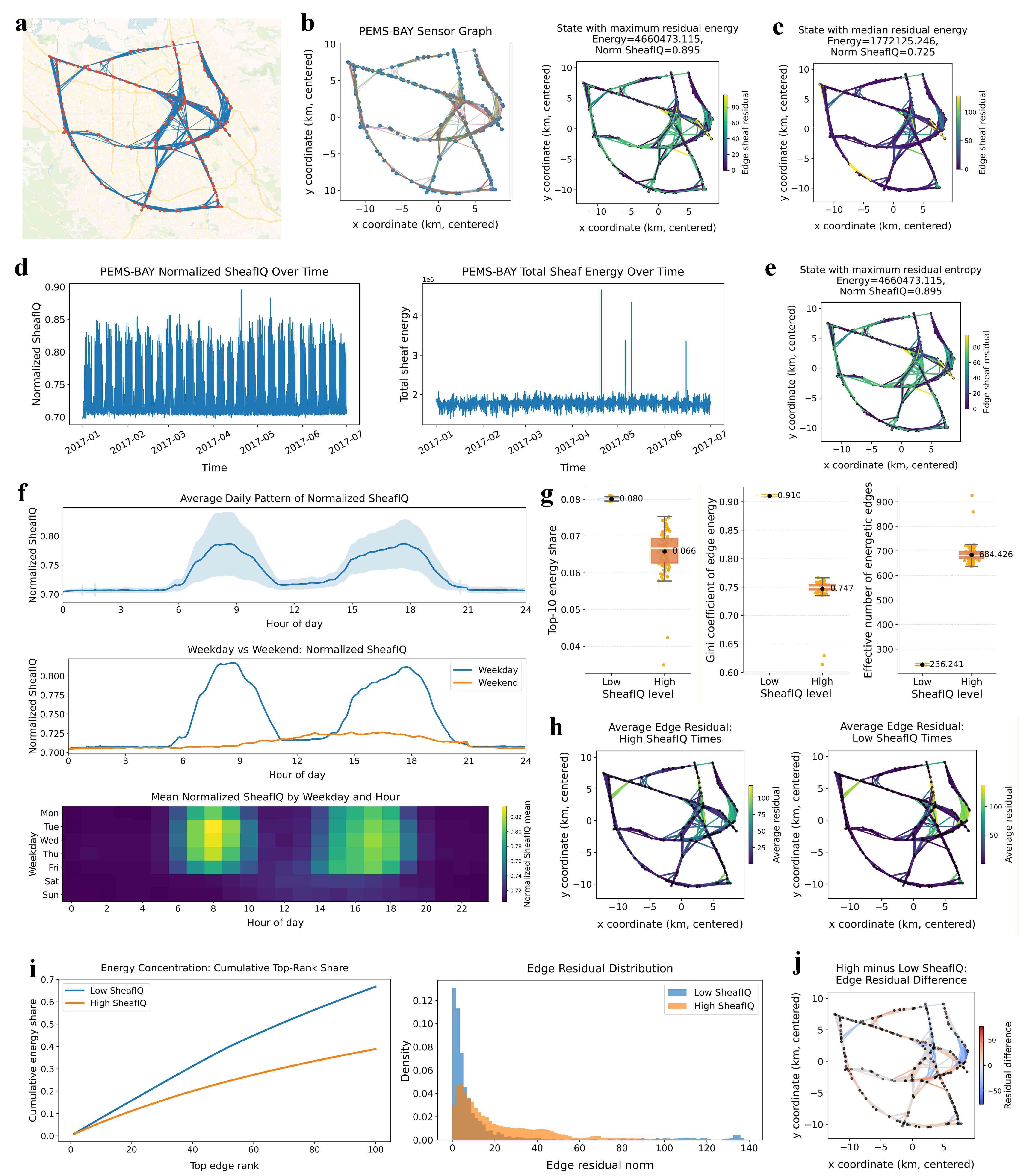}
	\caption{\textbf{Spatial organization analysis of SheafIQ on the PEMS-BAY traffic dataset.}
(a) Road network and sensor locations.
(b) Graph representation of the traffic network together with the traffic state exhibiting the maximum residual energy.
(c) Traffic state with median residual energy.
(d) Time series of normalized SheafIQ and total sheaf residual energy.
(e) Traffic state with maximum residual entropy.
(f) Temporal patterns of normalized SheafIQ, including the average daily profile, weekday--weekend comparison, and weekday--hour heatmap.
(g) Comparison of spatial organization metrics between high- and low-SheafIQ states, including the Top-10 energy share, Gini coefficient of edge energy, and effective number of energetic edges.
(h) Average edge residual maps for high- and low-SheafIQ states.
(i) Cumulative energy concentration curves and edge residual distributions for high- and low-SheafIQ states.
(j) Difference map of average edge residuals between high- and low-SheafIQ states.}
	\label{supfig:PEMS-BAY_1} 
\end{figure}

\begin{figure}[htbp]
	\centering
	\captionsetup{font=footnotesize, justification=justified,  singlelinecheck=false, skip=3pt}
	\includegraphics[width=0.95\textwidth]{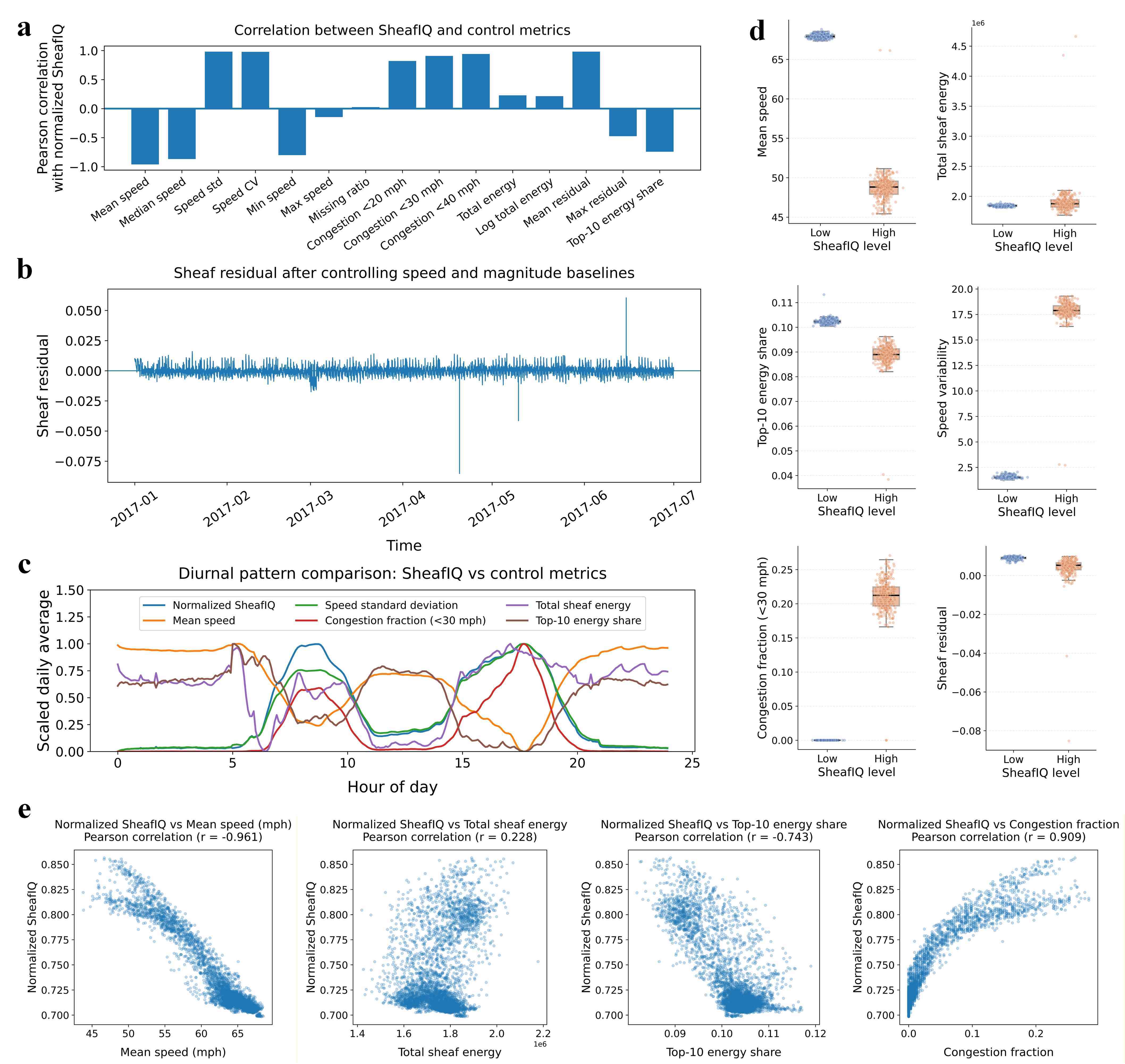}
	\caption{\textbf{Control-metric analyses of SheafIQ on the PEMS-BAY traffic network.}
(a) Pearson correlations between normalized SheafIQ and representative traffic and residual-based control metrics.
(b) Residual component of normalized SheafIQ after regressing out speed- and magnitude-related baseline variables.
(c) Comparison of average daily patterns between normalized SheafIQ and representative traffic metrics.
(d) Comparison of traffic statistics between high- and low-SheafIQ level.
(e) Scatter plots showing the relationships between normalized SheafIQ and representative control metrics.}
	\label{supfig:PEMS-BAY_2} 
\end{figure}

\begin{figure}[htbp]
	\centering
	\captionsetup{font=footnotesize, justification=justified,  singlelinecheck=false, skip=3pt}
	\includegraphics[width=0.7\textwidth]{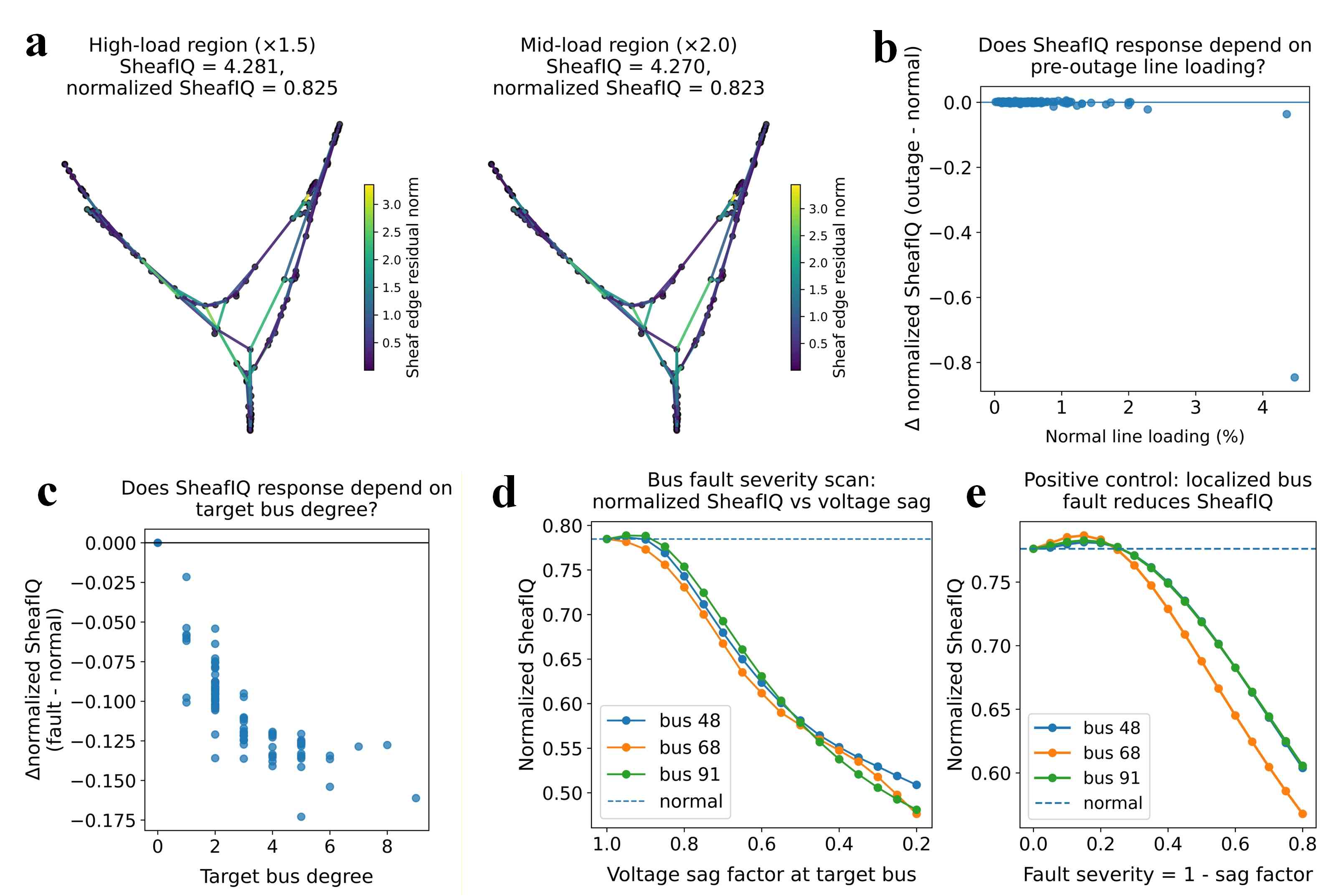}
	\caption{
\textbf{Additional analyses of SheafIQ on the IEEE 118-bus power grid.}
(a) Additional examples of topology-preserving regional load perturbations with corresponding edge residual distributions.
(b) Relationship between normalized SheafIQ changes induced by transmission-line outages and pre-outage line loading.
(c) Relationship between normalized SheafIQ changes induced by localized bus faults and target bus degree.
(d) Normalized SheafIQ under increasing fault severity for additional representative buses.
(e) Comparison between localized bus faults and regional load perturbations under representative operating conditions.
}
	\label{supfig:IEEE118_2} 
\end{figure}

\begin{figure}[htbp]
	\centering
	\captionsetup{font=footnotesize, justification=justified,  singlelinecheck=false, skip=3pt}
	\includegraphics[width=0.95\textwidth]{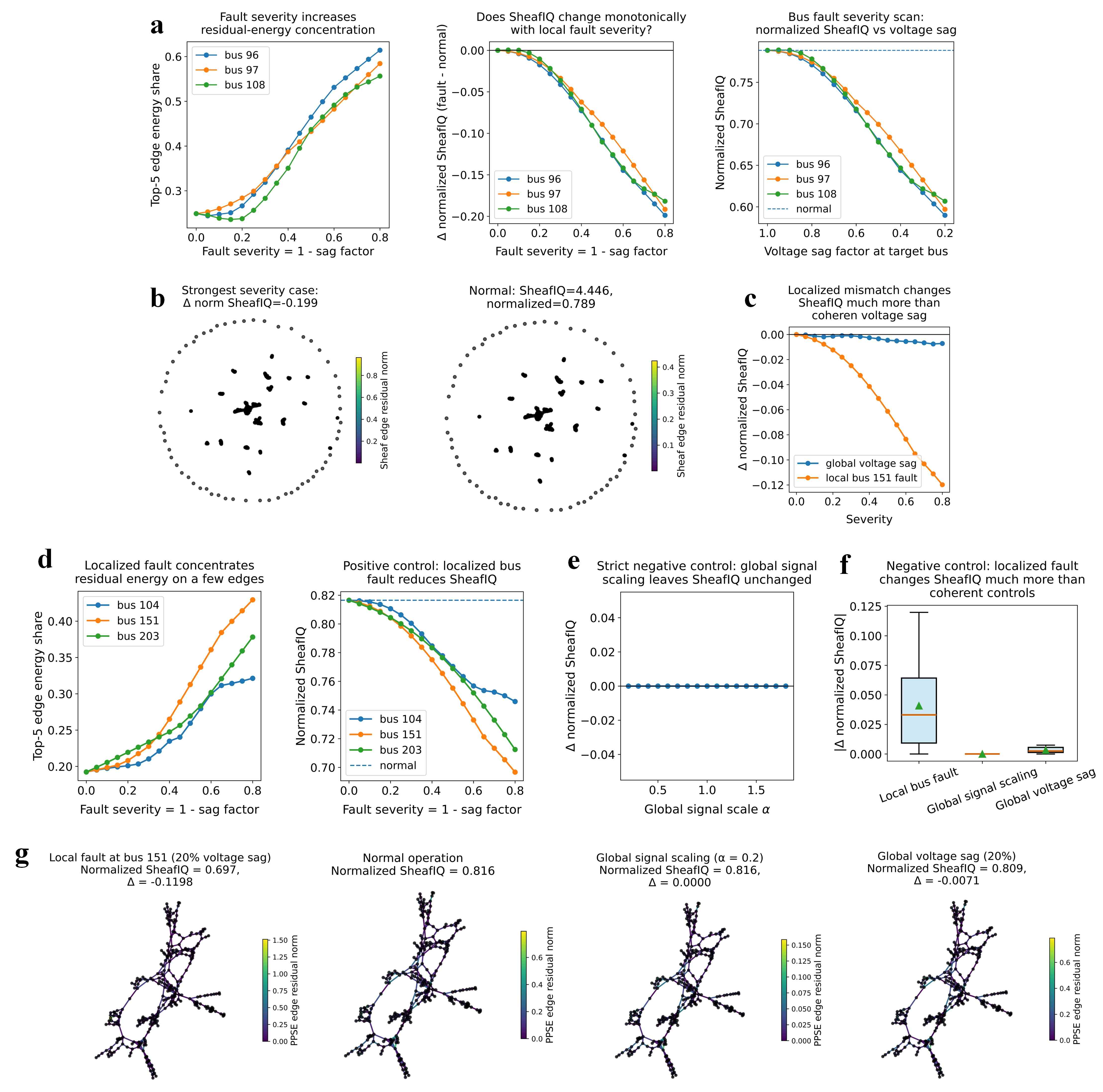}
	\caption{
\textbf{Generalization of SheafIQ analysis on the IEEE300 power grid.}
(a) Responses of normalized SheafIQ, the Top-5 edge energy share, and normalized SheafIQ under increasing localized fault severity for three representative buses.
(b) Comparison of edge residual distributions between normal operation and the strongest localized fault.
(c) Comparison between localized bus faults and coherent global voltage sag, showing that localized perturbations produce substantially larger changes in normalized SheafIQ.
(d) Changes in the Top-5 edge energy share and normalized SheafIQ under increasing localized fault severity for three additional representative buses.
(e) Strict negative control showing that normalized SheafIQ remains unchanged under coherent global signal scaling.
(f) Comparison of absolute normalized SheafIQ changes induced by localized bus faults, coherent global signal scaling, and coherent global voltage sag.
(g) Representative residual maps for localized bus fault, normal operation, coherent global signal scaling, and coherent global voltage sag.
}
	\label{supfig:IEEE300} 
\end{figure}

\end{document}